\documentclass[11pt,onecolumn,draftcls]{IEEEtran}

\usepackage[english]{babel}
\usepackage{algorithm} 
\usepackage{algorithmic} 
\usepackage{multirow}
\usepackage{graphicx}

\usepackage{amscd,amsmath,amssymb,amsthm}

\newtheorem{prop}{Proposition}
\usepackage[bookmarks]{hyperref}

\usepackage{caption2}

\begin{document}
	%
\title{Autonomous and Connected Intersection Crossing Traffic Management using Discrete-Time Occupancies Trajectory}

	%
	%
	
	\author{Qiang~Lu,
		and~Kyoung-Dae~Kim,~\IEEEmembership{Member,~IEEE}

		\thanks{Qiang Lu and Kyoung-Dae Kim are with the Department
			of Electrical and Computer Engineering, University of Denver, Denver,
			CO, 80210 USA e-mails: qiang.lu@du.edu, kyoung-dae.kim@du.edu.}
		}

	\maketitle
	
	\begin{abstract}
	In this paper, we address a problem of safe and efficient intersection crossing traffic management of autonomous and connected ground traffic. Toward this objective, we propose an algorithm that is called the Discrete-time occupancies trajectory based Intersection traffic Coordination Algorithm (DICA). We first prove that the basic DICA is deadlock free and also starvation free. Then, we show that the basic DICA has a computational complexity of $\mathcal{O}(n^2 L_m^3)$ where $n$ is the number of vehicles granted to cross an intersection and $L_m$ is the maximum length of intersection crossing routes. 
	To improve the overall computational efficiency of the algorithm, the basic DICA is enhanced by several computational approaches that are proposed in this paper. The enhanced algorithm has the computational complexity of $\mathcal{O}(n^2 L_m \log_2 L_m)$. The improved computational efficiency of the enhanced algorithm is validated through simulation using an open source traffic simulator, called the Simulation of Urban MObility (SUMO). The overall throughput as well as the computational efficiency of the enhanced algorithm are also compared with those of an optimized traffic light control.
	\end{abstract}
	
	\begin{IEEEkeywords}
		Autonomous Vehicles, Intelligent Intersection Management, Discrete-Time Occupancies Trajectory (DTOT), Computational Complexity
	\end{IEEEkeywords}

	%
	\IEEEpeerreviewmaketitle
	
\newpage
\listoffigures
\clearpage

\section{Introduction} \label{sec:intro}

Over the past decades, the vision for autonomous vehicles and autonomous ground traffic systems has indeed attracted a lot of attention and has catalyzed unprecedented research and development efforts from academia, industry, government, etc. Some examples are the California PATH Automated Highway System (AHS) program \cite{horowitz2000control} during the mid of 1990s and also the series of DARPA (Defense Advanced Research Projects Agency) Challenges \cite{darpa} that have happened during the 2000s. Many automobile companies are also investing huge amounts of money in developing their own self-driving vehicles or vehicles with many advanced driving assistance systems \cite{bengler2014three}. However, despite many recent successful road testing results of several self-driving cars such as Google driverless car \cite{markoff2010google}, it is hard to argue that the overall system-wide traffic safety as well as throughput will be improved substantially when we have a few autonomous vehicles among all other conventional vehicles. 
In fact, the potential of autonomous vehicles in terms of the traffic efficiency and safety will be unleashed when most cars on roads are autonomous and connected. 
Thus, in addition to many efforts to make today's traffic more efficient by improving utilization of traditional traffic infrastructure such as the work presented in \cite{chen}, we believe that it is also very important to develop traffic control algorithms that take advantages of connectivity and autonomy of autonomous vehicles to prepare for the next generation transportation system. 
However, while there have been many efforts toward this direction, the development of safe and efficient autonomous transportation systems is still at its early stage. In this paper, among many research problems like vehicle path planning \cite{konar}, autonomous parking control \cite{li2003autonomous}, collision avoidance \cite{mammeri,colombo}, relation between occupant experience and intersection capacity \cite{le2015autonomous}, intersection management of mixed traffic \cite{onieva2015multi} etc. that should be addressed toward this objective, we are particularly interested in addressing a problem of safe and efficient intersection crossing traffic management of autonomous connected traffic since intersections are certainly the most critical traffic environments from the perspective of safety as well as throughput.

In literature, there are a number of notable results for autonomous intersection crossing traffic management. 
In \cite{lee2012development}, Lee et al. proposed an algorithm, called the Cooperative Vehicle Intersection Control (CVIC), which manipulates every individual vehicle's driving motion by providing them proper acceleration or deceleration rate so that vehicles can cross the intersection safely. Wu et al. \cite{wu2012cooperative} introduced a new intersection traffic management framework that is formulated as a combinatorial optimization problem and solved the problem approximately using the ant colony system algorithm \cite{dorigo1996ant}. Most of them are centralized approaches in which control decisions are made typically by a central agent. Decentralized intersection control approaches have also been proposed in literature. For example, \cite{malikopoulos2016decentralized} formulated a decentralized framework whereby each autonomous vehicle minimizes its energy consumption under the throughput-maximizing timing constraints and hard safety constraints to avoid rear-end and lateral collisions. A complete analytical solution of the decentralized problems was presented in the paper. These approaches are similar in that they all ensure safety within an intersection by preventing vehicles with conflicting intersection crossing routes from being inside the intersection at the same time. 
To further improve the overall intersection crossing traffic throughput, some researchers eliminated this conservative restriction by discretizing an intersection space so that vehicles can exist simultaneously within an intersection but not within a same discretized space within the intersection. 
The representative approach is the reservation-based approach AIM (Autonomous Intersection Management) proposed in \cite{dresner2008multiagent}. 
In AIM, cars request and receive time slots from the intersection during which they may pass. Similar and improved approaches \cite{li2013modeling, jin2012advanced,wuthishuwong2015safe} were also proposed afterwards. For example, \cite{li2013modeling} proposed ASL (Advance Stop Location) concept which is a predefined advance stop location other than the traditional stop line at the entrance of an intersection for a vehicle with rejected reservation. The slow-reservation-speed issue which increases the total traversal time within the intersection could be improved by the ASL.
Representative centralized approaches also include auction-based intersection managements proposed in \cite{carlino2013auction, vasirani2012market}. A decentralized approach based on a vehicle-to-vehicle (V2V) coordination protocol was proposed in \cite{azimi2013reliable}. Roughly speaking, these approaches are all based on the grid cell partitioning of an intersection space. 
In \cite{dresner2008multiagent}, the effect of the grid cell granularity on the computational efficiency of an intersection traffic management framework such as AIM was studied. Clearly, higher granularity gives more flexibility for better traffic throughput. However, the computational complexity increases proportionally to the square of the granularity. On the other hand, when the cell size becomes large for better computational efficiency, one can see that the intersection space is not utilized efficiently resulting in lower traffic throughput. 
Therefore, to overcome this trade-off issue between the granularity and computational efficiency of an algorithm, it might be a good alternative approach to utilize each vehicle's actual occupancy instead of grid cells to improve the overall traffic throughput. And this has motivated our research on this topic.
Some other research works on autonomous intersection management can be found in \cite{wu2012cooperative, kim2014mpc, kowshik2011provable}. 

As an approach to address the above mentioned granularity issue, we proposed a novel intersection traffic management scheme in our earlier work \cite{lu2016intelligent} based on the idea of the \textit{Discrete-Time Occupancies Trajectory (DTOT)}.  Conceptually, a DTOT is a discrete-time sequence of a vehicle's actual occupancy within an intersection space. Hence, a DTOT-based intersection management scheme can utilize the intersection area much more efficiently than other grid partition based approaches. Furthermore, the proposed interaction mechanism between an intersection and vehicles  allows the flexibility that each vehicle can choose its path as well as motions along the path that a vehicle wants to take to cross an intersection. The management scheme is only dealing with head vehicles which reduces largely the communication needs for vehicles and the computational complexity of the central control agent.
In this paper, we provide an in-depth analysis of the original DTOT-based Intersection Coordination Algorithm (DICA) to show that it satisfies the liveness property in terms of deadlock as well as starvation issues and also to derive the overall computational complexity of the algorithm. Another contribution of this paper is that we propose several computational approaches to improve the overall computational efficiency of the DICA and also enhance the algorithm accordingly so that it can be operated in real-time for autonomous and connected intersection crossing traffic management. We also present simulation results that show the improved computational efficiency of the enhanced algorithm and the overall throughput performance in comparison with that of an optimized traffic light control.

The rest of this paper is organized as follows: 
In Section \ref{sec:dtot}, we introduce the main ideas, concepts, assumptions, notations, and also the basic algorithm, called DICA in short, developed for the DTOT-based intersection crossing traffic management. In Section \ref{sec:anal}, we show that the basic DICA is deadlock and also starvation free. In this section, we also discuss in detail about the computational complexity of the algorithm. Several approaches to improve the overall computational efficiency of the algorithm are discussed in Section \ref{sec:improv}. The overall computational efficiency as well as the throughput performance of the enhanced algorithm are evaluated through simulations in Section \ref{sec:sim}. Finally, Section \ref{sec:conc} concludes this paper.

\section{DTOT-based Intersection Traffic Management} \label{sec:dtot}
\begin{figure} [!t]
	\centering
	\includegraphics[width=4.in]{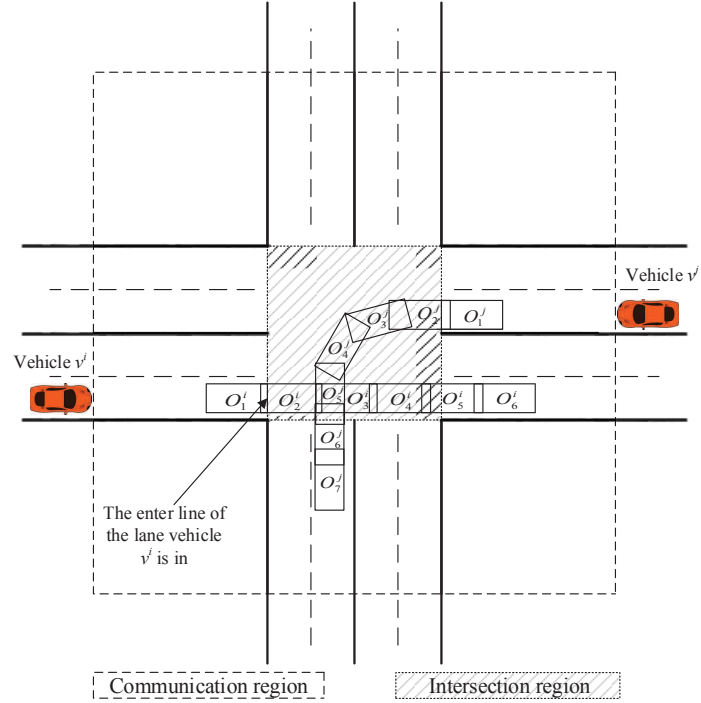}
	\caption[DTOTs of two conflicting vehicles]{DTOTs of two conflicting vehicles. ($O^p_q$ represents the $q$-th occupancy in a vehicle $v^p$'s DTOT. Note that occupancies in this figure are intentionally made very sparse for clear illustration purpose. DTOT starts with the occupancy in which the vehicle's front bumper first contacts the enter line of its lane of an intersection, and ends with the occupancy that the vehicle is completely out of the intersection region.)}
	\label{fig:example}
\end{figure}

In this section, we introduce the basic idea and algorithm of the DTOT-based intersection management scheme that is developed for autonomous and connected intersection crossing traffic in which all vehicles are autonomous vehicles (AVs) and capable of wireless vehicular communication. We assume that an intersection is also equipped with wireless communication capability as well as a computation unit so that it can exchange information with vehicles and perform necessary computations to coordinate vehicles to cross the intersection safely. At an intersection, there is no traffic light that controls the intersection crossing traffic. Instead, each vehicle communicates with the intersection, which we call the \emph{Intersection Control Agent} (ICA) from now on, to get permission to access the intersection. As shown in Figure \ref{fig:example}, an intersection consists of two regions. The bigger region in the figure, which we call the \emph{communication region}, is defined by the wireless vehicular communication range. The smaller region in the figure, which we call the \emph{intersection region}, is the area within an intersection that is shared by all roads connected to the intersection. We also assume that each vehicle is equipped with an RFID (Radio Frequency IDentification) chip and there are detectors installed at the entrance of the communication region so that ICA can detect each vehicle's identification number (VIN), the lane on which a vehicle is approaching an intersection, and the time when a vehicle enters the communication region. Since all vehicles are autonomous, we assume that each vehicle can obtain its position, speed, and the relative distance to an intersection precisely and also can avoid collisions with other vehicles autonomously when it is approaching an intersection. With regard to wireless vehicular connectivity, we only require information exchange between AVs and ICA. Thus, there is no V2V communication.   
Since the focus of this paper is to develop an algorithm for ICA for safer and higher throughput intersection crossing traffic, we simply assume that we have an ideal wireless vehicular communication performance such that all data packages are exchanged correctly and timely. However, it is important to note that, despite such an ideal communication assumption, our DTOT-based algorithm can still be applicable in practice with small modifications of the algorithm to take into account the communication unreliability thanks to the above mentioned information collection mechanism through detectors.

\subsection{Interaction between ICA and AV}
In the autonomous and connected intersection traffic considered in this paper, an AV and ICA start to interact with each other by exchanging some messages through vehicular wireless communication when the AV enters into the communication range of ICA. As shown in Figure \ref{fig:stateMachine}, an AV and ICA exchange some specific types of messages for their intersection crossing coordination. The interaction is initiated from an AV by sending a message, called a REQUEST, to reserve a sequence of space and time to cross the intersection. Each REQUEST message contains information that is necessary for a vehicle's space-time reservation for its intersection crossing such as (i) the VIN, (ii) the Vehicle Size (VS) that is simply a vehicle's length and width, and (iii) a vehicle's discrete time state trajectory, which we call the Timed State Sequence (TSS), starting from the entrance of an intersection region to the moment when the vehicle crosses the intersection region completely. Note that it is implicitly assumed that each discrete time state of a vehicle in TSS is also timed. This means that if a vehicle state $\mathbf{x}_t$ is given, then we can say that a vehicle possesses the state $\mathbf{x}$ at time $t$. For simplicity of our discussion, we simply assume that the state $\mathbf{x}$ of a vehicle consists of the $(x, y)$ coordinate of the vehicle's location and the orientation $\theta$. We also assume that, while it is possible that each vehicle can have different sampling period to generate its TSS, all vehicles use the same sampling period which is small enough to generate a close approximation of the vehicle's actual continuous motion within an intersection.

\begin{figure}[!htb]
	\centering
	\includegraphics[width=3.8in]{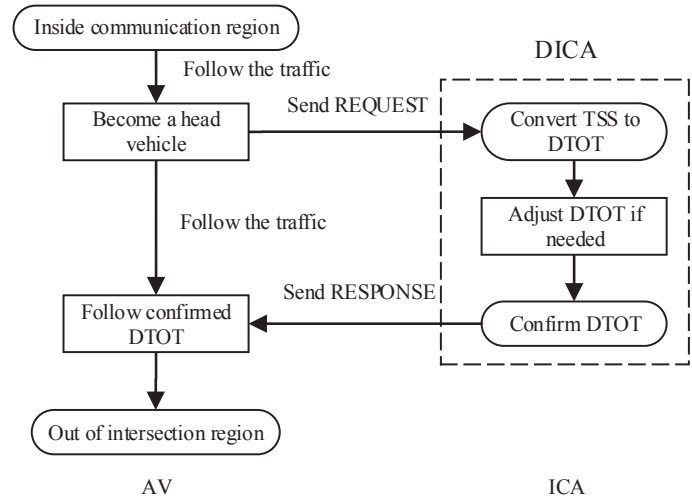}
	\caption{Interaction between an AV and the ICA.}
	\label{fig:stateMachine}
\end{figure}

As shown in Figure \ref{fig:stateMachine}, an AV will send REQUEST message to ICA only when it becomes a \emph{head vehicle}. A vehicle is considered a head vehicle on its lane if either there are no vehicles in front of it or the vehicle which is immediately in front of it has begun to enter the intersection region. Note that ICA also knows whether a vehicle is a head vehicle or not according to the list of vehicles for each lane. Thus, a REQUEST message not from a head vehicle will be neglected by ICA. The list can be constructed in ICA since, as explained earlier, ICA knows each vehicle's VIN, the lane on which the vehicle is approaching, and the time when a vehicle passes a detector installed at the boundary of the communication region of an intersection.
To respond to a REQUEST message from a head vehicle, ICA first converts the TSS to the corresponding DTOT using the VS information which is also contained in the received REQUEST message. Based on the length and width of a vehicle consisted in the VS, a DTOT is simply a sequence of timed rectangular spaces that a vehicle needs to occupy within an intersection region to cross the intersection.
Now, ICA uses DTOT to determine whether the requested DTOT can be accepted or not. If ICA finds any potential risk of collision with the request, then it adjusts the requested DTOT slightly in order to eliminate collisions. And then ICA confirms the adjusted collision-free DTOT and sends it back to the vehicle using a RESPONSE message so that the vehicle can follow the confirmed DTOT to cross the intersection. Note that when ICA sends a RESPONSE message to a vehicle, it actually sends the confirmed TSS not the confirmed DTOT. More detailed explanation on how to process the requested TSS to generate a confirmed DTOT is presented in the following section. In the sequel, we say that a vehicle is a \emph{confirmed vehicle} if it has received a confirmed DTOT from ICA. And we assume that every vehicle is able to follow the confirmed DTOT precisely.

\begin{algorithm}[!thb]   
	\caption{DICA (\textbf{D}TOT-based \textbf{I}ntersection traffic \textbf{C}oordination \textbf{A}lgorithm)}
	\begin{algorithmic}[1]
		\STATE Let $\mathcal{S}$ be the set of confirmed vehicles and $n = \vert \mathcal{S} \vert$.
		\STATE Let $v^i$ be the vehicle to be considered for confirmation.
		\STATE Convert $TSS(v^i)$ to $DTOT(v^i)$ \label{convert}
		\STATE Call checkFV$(\mathcal{S}, DTOT(v^i)) \rightarrow DTOT(v^i)$ \label{FV}
		\STATE Call getCV$(\mathcal{S}, DTOT(v^i)) \rightarrow \mathcal{C}$ \label{getCV}
		\WHILE {$\mathcal{C} \neq \emptyset$} \label{while}
		\STATE Pop the first vehicle in $\mathcal{C} \rightarrow v^j$
		\STATE Call updateDTOT$(DTOT(v^i), DTOT(v^j)) \rightarrow DTOT(v^i)$ \label{update}
		\STATE Call getCV$(\mathcal{S}, DTOT(v^i)) \rightarrow \mathcal{C}$
		\ENDWHILE \label{endwhile}
		
		\STATE Store $DTOT(v^i)$ for vehicle $v^i$ \label{store}
		\STATE Convert $DTOT(v^i)$ to $TSS(v^i)$
		\STATE Send $TSS(v^i)$ to vehicle $v^i$ \label{send}
		
	\end{algorithmic}
	\label{code:main}
\end{algorithm}

\subsection{DTOT-based Intersection Traffic Coordination}
ICA processes a REQUEST message from a head vehicle according to the procedures shown in Algorithm \ref{code:main} which we call the DTOT-based Intersection traffic Coordination Algorithm (DICA). As shown in the algorithm, DICA uses a few sets and notations. We use TSS($v$) to denote a TSS and DTOT($v$) to denote a DTOT for a vehicle $v$ respectively.  We also use $\mathcal{S}$ to denote the set of vehicles which have already been confirmed at the time when a REQUEST message is being processed. We say that two vehicles are  \emph{space-time conflicting} if their trajectories are conflicting not only in space but also in time. More precisely, two vehicles are considered to be in space-time conflict in our algorithm when their DTOTs have at least one pair of occupancies that are conflict in both space and time. We use another set $\mathcal{C}$ in Algorithm \ref{code:main} to represent the subset of $\mathcal{S}$ which contains the set of vehicles whose confirmed DTOTs have space-time conflict with the DTOT of the vehicle that is currently being processed for confirmation. Vehicles in $\mathcal{C}$ are ordered in ascending order of a certain attribute of their confirmed DTOTs. To explain this attribute more clearly, let us consider a situation when DICA processes a vehicle $v^i$'s DTOT and there are two vehicles $v^j$ and $v^k$ in the set $\mathcal{C}$. Now let us suppose that DTOT($v^j$) starts to space-time conflict with DTOT($v^i$) from its $n$-th occupancy and DTOT($v^k$) starts to space-time conflict with DTOT($v^i$) from its $m$-th occupancy. If we use $O^p_{q}$ to denote the $q$-th occupancy within DTOT($v^p$) and $\tau(O^p_{q})$ be the time when the vehicle $v^p$ occupies $O^p_q$, then we say that, in this particular situation, $\tau(O^j_n)$ is the first time at which $v^j$ starts to collide with $v^i$. Similarly, $\tau(O^k_m)$ is the time at which $v^k$ starts to collide with $v^i$. In the sequel, this specific time instant for each vehicle in $\mathcal{C}$ is represented by the variable `\emph{firstTimeAtCollision}'.  
In this particular situation, $\tau(O^j_n)$ and $\tau(O^k_m)$ are denoted by $v^j.firstTimeAtCollision$ and $v^k.firstTimeAtCollision$, respectively. Vehicles in the set $\mathcal{C}$ are ordered according to this variable. Specifically, if $v^j.firstTimeAtCollision$ is earlier than $v^k.firstTimeAtCollision$, then $v^j$ gets higher priority than $v^k$ and vice versa. 
To see more clearly how the `\emph{firstTimeAtCollision}' is determined, we can consider an illustrative example shown in Figure \ref{fig:example}. In the figure, DTOT($v^i$) and DTOT($v^j$) have space conflicts in \{$O^i_2, O^i_3$\} and \{$O^j_5, O^j_6,$\}. If we assume that these occupancies are also conflicting in time, then $v^j.firstTimeAtCollision$ with respect to the vehicle $v^i$ is $\tau(O^j_5)$. 

As shown in Algorithm \ref{code:main}, when ICA receives a REQUEST message from a head vehicle $v^i$, it first converts the TSS($v^i$) into the corresponding DTOT($v^i$) using the vehicle's VS. Then ICA calls the function \verb+checkFV()+ to determine if there exist \emph{front vehicles} (See Section \ref{subsubsec:front} for more details about front vehicles.) that affect the vehicle $v^i$'s motion and also to adjust $v^i$'s DTOT if needed. Then the function \verb+getCV()+ is called to determine the set $\mathcal{C}$ which is the set of vehicles whose DTOTs are space-time conflicting with DTOT($v^i$). The \verb+updateDTOT()+ function adjusts DTOT($v^i$) appropriately so that DTOT($v^i$) avoids space-time conflict with other vehicle's DTOT. These two functions are iteratively called within the \verb+while+ loop until the set $\mathcal{C}$ becomes empty, which indicates that no vehicles in the set $\mathcal{C}$ will collide with the vehicle $v^i$. After DTOT($v^i$) is appropriately adjusted and confirmed that there is no space-time conflict with all other confirmed vehicles, then the confirmed DTOT($v^i$) is converted into TSS($v^i$). Finally, ICA sends the confirmed TSS($v^i$) back to the vehicle $v^i$ so that the vehicle can cross the intersection safely by following the confirmed DTOT.
In the following sections, we provide more detailed explanation on several functions called within DICA.

\subsubsection{Collision Avoidance with Front Vehicles} \label{subsubsec:front}

As shown in Figure \ref{fig:front}, there are two types of front vehicles when a vehicle $v^i$ is approaching and crossing an intersection. 
In DICA, a vehicle is considered as a front vehicle of $v^i$ if the vehicle comes from another lane but has the same exit lane as vehicle $v^i$ or the vehicle is immediately in front of $v^i$ and has the exact same intersection crossing route as that of $v^i$. For a vehicle $v^i$, if there is another confirmed vehicle whose exit lane is the same as that of vehicle $v^i$ and will exit the intersection earlier, then they may collide immediately after crossing the intersection if the speed of vehicle $v^i$ is higher than that of the other confirmed vehicle. To address this problem, AIM \cite{dresner2008multiagent} adopted a simple strategy which gives one second separation time between these two vehicles. However, it is important to note that the separation time should depend on the speeds of the two vehicles. Hence, instead of using a fixed separation time approach, we use an approach that restricts the maximum speed of a following vehicle by the speed of the front vehicle. In the example situation (a) shown in Figure \ref{fig:front}, the vehicle $v^i$'s maximum allowed speed within an intersection is restricted by the front vehicle's exit speed. If there is another confirmed vehicle that has the same intersection crossing route as vehicle $v^i$, we adjust $v^i$'s speed to leave adequate distance between them. In Algorithm \ref{code:main}, the function \verb+checkFV()+ looks for the existence of above mentioned front vehicles from all confirmed vehicles and delay the new head vehicle to avoid potential collisions if needed.

\begin{figure}[!t]
	\centering
	\includegraphics[width=2.8in]{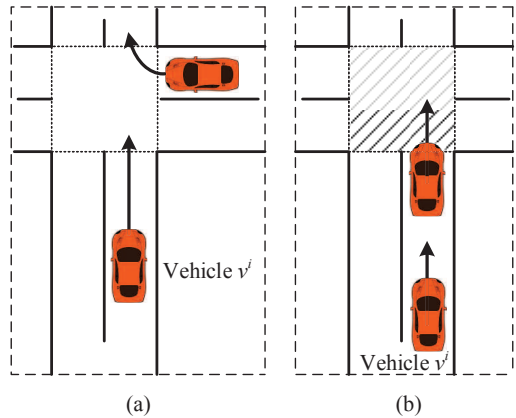}
	\caption[Example situations of front vehicles]{Example situations of front vehicles: (a) vehicles with different routes but same exit lane, and (b) vehicles with same intersection crossing routes.}
	\label{fig:front}
\end{figure}

\subsubsection{Vehicles for Collision Avoidance} \label{subsubsec:oti}
The function \verb+getCV()+ returns the set $\mathcal{C}$ that contains vehicles which will cause potential collisions inside the intersection with vehicle $v^i$. To better understand the operation of function \verb+getCV()+, it is necessary to introduce the way we check the space-time conflict between two occupancies from DTOTs of two vehicles. 
For every individual occupancy in a DTOT of a vehicle, we define the entrance time ($\tau_{lb}$) and the exit time ($\tau_{ub}$) of the occupancy as the times when the vehicle first contacts and is totally out of the occupancy. These two times can be estimated by taking the times of the previous and next occupancies which are the closest to the occupancy while having no overlapping area. As an example, for the occupancy $O^j_4$ of the vehicle $v^j$ in Figure \ref{fig:example}, the entrance time $\tau_{lb}(O^j_4)$ and the exit time $\tau_{ub}(O^j_4)$ of that occupancy can be determined by $\tau(O^j_2)$ and $\tau(O^j_6)$, respectively. Note that a DTOT for a vehicle consists of many more numbers of occupancies in practice. Hence, the entrance times and exit times determined in this way can be very close to the actual entrance and exit times of the occupancy. For the first several occupancies in a DTOT, there may not be a previous occupancy that has no overlapping area with themselves. For these occupancies, we simply take the first occupancy's time in the DTOT as these occupancies' entrance time. As an example shown in Figure \ref{fig:example}, we use $\tau(O^j_1)$ as the entrance time $\tau_{lb}(O^j_2)$ for the occupancy $O^j_2$. Similarly, we take the last occupancy's time as the exit time $\tau_{ub}$ for the last several occupancies in a DTOT.

As shown in Algorithm \ref{code:GetCV}, the function \verb+getCV()+ determines the set $\mathcal{C}$ by checking space-time conflict for every pair of occupancies $(O^i_n, O^j_m)$ for all $n, m$, and $j$ in the set $\mathcal{S}$. Since an occupancy in a DTOT is represented as a rectangle, it is relatively straightforward to do space conflict checking. For this, Algorithm \ref{code:GetCV} simply checks if two rectangles have non-empty intersection or not. 
If a pair of occupancies $(O^i_n, O^j_m)$ are space-conflicting, then the function continues to investigate these occupancies to determine if they are in time-conflict as well. The above explained entrance and exit times of an occupancy are used for this purpose. For a given occupancy $O$, the function \verb+getOTI()+ calculates these entrance $\tau_{lb}(O)$ and exit $\tau_{ub}(O)$ times for that occupancy and returns a corresponding time interval $I(O) := [\tau_{lb}(O), \tau_{ub}(O)]$ which we call the \emph{occupancy time interval} in the sequel. 
Then the two occupancy time intervals for the pair of space-conflicting occupancies are compared to determine if these occupancies are also occupied around the same time. 
If a pair of occupancies $(O^i_n, O^j_m)$ are conflicting in both space and time, then the vehicle $v^j$ is included in the set $\mathcal{C}$ and the corresponding \emph{firstTimeAtCollision} is determined so that the vehicle $v^j$ is appropriately ordered within the set $\mathcal{C}$.

\begin{algorithm}[!htb]    
	\caption{getCV$(\mathcal{S}, DTOT(v^i))$}
	\begin{algorithmic}[1]
		
		\STATE $\mathcal{C} = \emptyset$
		\FOR {$v^j$ in $\mathcal{S}$}
		\FOR {$O^j_{k_j}$ in $DTOT(v^j)$} \label{Ov}
		\IF {$v^j$ not in $\mathcal{C}$} \label{CVi}
		\FOR {$O^i_{k_i}$ in $DTOT(v^i)$}
		\IF {$O^j_{k_j} \cap O^i_{k_i} \neq \emptyset$} \label{if}
		\STATE Call getOTI($O^j_{k_j})$ $\rightarrow I(O^j_{k_j}) := [\tau_{lb}({O}^j_{k_j}), \tau_{ub}({O}^j_{k_j})]$
		\STATE Call getOTI($O^i_{k_i}$) $\rightarrow I(O^i_{k_i}) := [\tau_{lb}({O}^i_{k_i}), \tau_{ub}({O}^i_{k_i})]$
		\IF {$I({O}^j_{k_j}) \cap I({O}^i_{k_i}) \neq \emptyset$}
		\STATE Assign $\tau_{lb}({O}^j_{k_j}) \rightarrow v^j.firstTimeAtCollision$
		\STATE Push $v^j$ into $\mathcal{C}$
		\STATE Sort $\mathcal{C}$ in ascending order of \emph{firstTimeAtCollision}
		\ENDIF 
		\ENDIF \label{endif}
		\ENDFOR
		\ENDIF
		\ENDFOR
		\ENDFOR
		
	\end{algorithmic}
	\label{code:GetCV}
\end{algorithm}

\subsubsection{DTOT Update} \label{subsubsec:update}
The first vehicle $v$ in the set $\mathcal{C}$ is the earliest vehicle that is space-time conflicting with vehicle $v^i$. Then, in line \ref{update} of Algorithm \ref{code:main}, the function \verb+updateDTOT()+ modifies vehicle $v^i$'s DTOT to avoid collision with vehicle $v$ based on space-time conflicting occupancies between vehicles $v^i$ and $v$. However, it is still uncertain whether $\mathcal{C}$ will be empty or not after this update of avoiding collision with vehicle $v$. 
In fact, it is still possible that the modified DTOT of vehicle $v^i$ will be in space-time conflict with DTOTs of other confirmed vehicles.
Hence, to ensure that vehicle $v^i$ avoids collision with all other confirmed vehicles, it is necessary to construct $\mathcal{C}$ based on the updated vehicle $v^i$'s DTOT and update the DTOT again to avoid collision with the first vehicle in the set.
This process is repeated in the \verb+while+ loop in Algorithm \ref{code:main} until the set $\mathcal{C}$ becomes empty which means that vehicle $v^i$ is not conflicting with any confirmed vehicles.
Our current strategy for updating a vehicle's DTOT is to delay the vehicle until other confirmed vehicles cross an intersection safely. While it is an interesting future research problem to develop more sophisticated approaches to improve the overall performance, the current simple delay strategy is still very effective to ensure collision free intersection traffic. Note that, since the times of occupancies in a vehicle's DTOT are always delayed whenever the vehicle's DTOT is updated, it is guaranteed that the vehicle can always meet the updated DTOT by simply decelerating to experience a longer time before entering the intersection. The worst case is that a vehicle may need to stop and wait for some time before an intersection to meet the given confirmed TSS from ICA.

\section{Analysis} \label{sec:anal}

\subsection{Liveness}
A \emph{deadlock} is a situation where two or more processes are unable to proceed and each process is waiting for another one to finish because they are competing for shared resources. In an intersection crossing traffic, a deadlock could happen when several vehicles are trying to cross the intersection at the same time. For example, if the coordination between vehicles who want to cross an intersection is not done appropriately, then a deadlock may occur between two vehicles on a same lane. As discussed in \cite{dresner2008multiagent}, it is possible that even when the vehicle in front cannot get confirmed due to the conflict of its intersection crossing route with those of other vehicles which are already confirmed to enter and cross an intersection, the vehicle in the back may get confirmed because its intersection crossing route is not conflicting with other confirmed vehicles' crossing routes. And the vehicle successfully reserves the space for its intersection crossing route within an intersection. In this situation, the front vehicle cannot get confirmed since some part of the intersection crossing route of it conflicts with that of the behind vehicle which is already confirmed and also the behind vehicle cannot proceed to cross the intersection due to the unconfirmed front vehicle.
A deadlock situation may also occur when several vehicles from different directions want to cross an intersection at the same time. This type of deadlock situation is discussed in detail in \cite{azimi2013reliable} for the case of four vehicles in which none of the vehicles can progress inside the intersection because each of the vehicles' next occupancies are already occupied by other vehicles. Now we show that DICA shown in Algorithm \ref{code:main} are free from these deadlock situations.

\begin{prop} DICA is deadlock free.
\end{prop}
\begin{proof}
Let $\mathcal{S}_k$ denote the set of confirmed vehicles at the $k$-th time step of DICA. Then, we show that the set $\mathcal{S}_k$ is deadlock free for all $k = 0, 1, 2, \cdots$ by induction. First, at time step $k = 0$, it is easy to see that there is no deadlock in $\mathcal{S}_0$ since no vehicle is confirmed yet, i.e., $\vert \mathcal{S}_0 \vert = 0$ where $\vert \cdot \vert$ denotes the cardinality of a set. Then, at time step $k > 0$, let us suppose that $\mathcal{S}_k$ is deadlock free and a new head vehicle $v^i$ is under consideration for confirmation. Note that, as discussed in Section \ref{sec:dtot}, a vehicle is considered by DICA for confirmation only if it is the head vehicle on its lane. Hence, it is trivial to see that there won't be a deadlock situation between the vehicle $v^i$ and other vehicle $v^{i'}$ which is behind $v^i$ since $v^{i'} \not\in \mathcal{S}_k$.
Next, let us note that once a vehicle $v^j$ is in $\mathcal{S}_k$, then the vehicle's DTOT will not be changed while and after a new vehicle $v^i$ is processed to be confirmed by DICA. 
Hence, it is easy to see that any vehicle which is in $\mathcal{S}_k$ at time step $k$ remains deadlock free at the next time step $(k+1)$. Now suppose that the new vehicle $v^i$ has been confirmed by DICA at time step $k$ and included in the set of confirmed vehicle at time step $(k+1)$, i.e., $v^i \in \mathcal{S}_{k+1}= \mathcal{S}_k \cup \{ v^i \}$. Since all vehicles in $\mathcal{S}_k \subset \mathcal{S}_{k+1}$ are deadlock free, if the new vehicle $v^i$ is deadlock free, then we know that $\mathcal{S}_{k+1}$ is deadlock free and this proves the deadlock free property of DICA. In fact, it is straightforward to see that $v^i$ is also deadlock free after its DTOT is updated and confirmed by DICA. First, note that modification of the vehicle $v^i$'s DTOT is not affected by any vehicle $v \not\in \mathcal{S}_k$. Instead, it is affected only by vehicles which are already in the $\mathcal{S}_k$. Since all vehicles in $\mathcal{S}_k$ are deadlock free and eventually proceed to cross and exit the intersection, the vehicle $v^i$'s DTOT is also updated so that the vehicle $v^i$ will eventually enter and cross the intersection while all vehicles in $\mathcal{S}_k$ cross the intersection safely. Thus, the vehicle $v^i$ is also deadlock free at time step $(k+1)$ and this concludes the proof of this proposition.
\end{proof}

In an intersection crossing traffic, a \emph{starvation} situation may occur when vehicles from a certain direction are waiting for a very long time or even indefinitely to be allowed to enter and cross an intersection while vehicles from other directions are continuously allowed to cross the intersection. Now we show that a starvation situation will not occur in an intersection crossing traffic that is coordinated by DICA.

\begin{prop} DICA is starvation free.
\end{prop}
\begin{proof}
First, let us recall that, as discussed in Section \ref{sec:dtot}, DICA considers a vehicle for confirmation only when the vehicle becomes the head vehicle on its lane. Now let $\sigma(v)$ be the vehicle $v$'s entrance time to the communication region of an intersection, $\mathcal{H}$ be the set of head vehicles which is ordered by $\sigma(v)$ for all $v \in \mathcal{H}$, and $\mathcal{H}^-$ be the set of vehicles which are approaching to cross an intersection but not included in the set $\mathcal{H}$. Clearly, $\vert \mathcal{H} \vert$ is bounded by the number of all lanes from which vehicles are approaching an intersection to cross and $\vert \mathcal{H}^- \vert$ is also bounded by both the number of lanes and the length of lanes within the communication region of an intersection.
Note that DICA processes vehicles in $\mathcal{H}$ for confirmation according to the order of vehicles in $\mathcal{H}$. Once the first vehicle in $\mathcal{H}$ is processed and gets confirmed, then the vehicle is removed from $\mathcal{H}$. Note that if DICA is not starvation free, then there must exist at least one vehicle $v \in \mathcal{H}$ such that the vehicle $v$ will never (or at least take an unnecessarily very long time to) become the first element in the ordered set $\mathcal{H}$. Thus, to prove the starvation free property of DICA, it suffices to show that, for any vehicle $v \in \mathcal{H}$, the vehicle $v$ will be removed from $\mathcal{H}$ in finite time.
To show this, we can consider the last vehicle $v_{last}$ in the ordered set $\mathcal{H}$. If $\sigma(v_{last}) \le \sigma(v)$ for all $v \in \mathcal{H}^-$, then the vehicle $v_{last}$ will be cleared right after all other vehicles in $\mathcal{H}$ are confirmed and this is the earliest time for $v_{last}$ to be removed from $\mathcal{H}$. On the other hand, if $\sigma(v_{last}) > \sigma(v)$ for all $v \in \mathcal{H}^-$ as the worst situation for $v_{last}$, then the vehicle $v_{last}$ might need to wait until all ($\vert \mathcal{H} \vert + \vert \mathcal{H}^- \vert$) vehicles get confirmed to be considered for confirmation. Thus, it is clear that the vehicle $v_{last}$ will be cleared from $\mathcal{H}$ in finite time. 
\end{proof}

\subsection{Computational Complexity} 
In this section, we analyze the computational complexity of DICA shown in Algorithm \ref{code:main}. 
Recall that $\mathcal{S}$ is the set of vehicles within the communication region of an intersection that has been confirmed to cross.
Let us assume that there are $n$ vehicles in $\mathcal{S}$, i.e., $\vert \mathcal{S} \vert = n$. 
Then we have the following result on the computational complexity analysis of DICA.

\begin{prop} \label{prop:dica}
DICA has $\mathcal{O}(n^2 L_m^3)$ computational complexity where $L_m$ is the maximum length of intersection crossing routes in an intersection.
\end{prop}
\begin{proof}
Let $v^i$ be the vehicle which is currently being processed by ICA for intersection crossing confirmation. 
Also let $N_m := \max_{k \in \mathcal{S}'} N^k$ where $\mathcal{S}' = \mathcal{S} \cup \{ v_i \}$ and $N^k$ is the number of occupancies in the vehicle $k$'s DTOT.
Then, in line 3 (Algorithm \ref{code:main}), it is easy to see that creating DTOT from the TSS and vehicle size information in the vehicle $v^i$'s REQUEST message involves only $\mathcal{O}(N_m)$ computational complexity.
In line 4 (Algorithm \ref{code:main}), as explained in Section \ref{sec:dtot}, the front vehicle checking function \verb+checkFV()+ does a simple comparison with every confirmed vehicle in $\mathcal{S}$ to see if there are any vehicles which might affect the vehicle $v^i$'s DTOT and modifies the DTOT if it is necessary to ensure enough separation time and distance between the vehicle $v^i$ and other vehicles in front. And this process requires computational complexities $\mathcal{O}(nN_m)$. 
Then, in line 5 (Algorithm \ref{code:main}), the function \verb+getCV()+ is called to identify the set of vehicles $\mathcal{C}$ in $\mathcal{S}$ whose DTOTs might be in space-time conflict with the vehicle $v^i$'s DTOT. (Note that, as shown in Algorithm \ref{code:GetCV}, $\mathcal{C}$ is an ordered set according to time of collision and it is clearly $\mathcal{C} \subseteq \mathcal{S}$.) Thus, to return the set $\mathcal{C}$ from the set $\mathcal{S}$, this function performs $n$ times of space-time conflict checking between the vehicle $v^i$ and vehicles in $\mathcal{S}$. If a nonempty set $\mathcal{C}$ is returned in line 5 (Algorithm \ref{code:main}), then, in lines 6 $\sim$ 10 (Algorithm \ref{code:main}), the vehicle $v^i$'s DTOT is iteratively updated until the set $\mathcal{C}$ becomes empty within the \verb+while+ loop. (As one can see in Algorithms \ref{code:main} and \ref{code:GetCV}, these steps are indeed the main part of the DICA algorithm and involve some computationally expensive operations. Hence, we describe the computational complexity of steps within the \verb+while+ loop separately in the next paragraph.)  
After the \verb+while+ loop, as the last steps in Algorithm \ref{code:main} in lines from \ref{store} to \ref{send}, the space-time conflict free DTOT for the vehicle $v^i$ is stored, converted into TSS, and then sent to $v^i$ so that the vehicle can cross the intersection according to the DTOT. Clearly, these steps are fairly simple in terms of computation and in fact require $\mathcal{O}(1)$ complexity.
Next, we analyze the computational complexity of steps within the \verb+while+ loop. 

\emph{Space-time conflict checking steps}: 
As described in Section \ref{sec:dtot}, space-time conflict checking in \verb+getCV()+ is done using DTOTs of vehicles. Specifically, the two nested \verb+if+ blocks from line 6 to line 14 in Algorithm \ref{code:GetCV} perform this operation. For space conflict checking, it is checked if there exists nonempty intersections between two occupancies: one from DTOT of the vehicle $v^i$ and another from DTOT of one of the vehicles in the set $\mathcal{S}$. This is done in the outer \verb+if+ block and requires $n \cdot N_m^2$ times of iteration in the worst case. If two vehicles have a space conflict, then Algorithm \ref{code:GetCV} proceeds to check for time conflict. To check time overlapping between two space conflicting occupancies, the function needs to calculate time intervals for these occupancies during which each vehicle occupies its occupancy. This can be done easily by comparing occupancy time between occupancies within the same DTOT. As an example, for a given occupancy $O^i_k$ which is the $k$-th occupancy within the vehicle $v^i$'s DTOT, the lower and upper bounds for the occupancy time can be determined by space overlapping checking between the occupancies $O^i_k$ and $O^i_{k'}$ for $k' = \{ 1, \cdots, N_m \} \setminus k$. Thus, the two function calls to \verb+getOTI()+ within the \verb+if+ block involve the computational complexity of $\mathcal{O}(N_m)$. Once the occupancy time intervals are determined, it is a straightforward calculation to check time overlapping as shown in line 9 of Algorithm \ref{code:GetCV} and it takes $\mathcal{O}(1)$ computational complexity. After identifying all space-time conflicting vehicles from the set $\mathcal{S}$ and storing them to the set $\mathcal{C}$, Algorithm \ref{code:GetCV} then sorts the set $\mathcal{C}$ according to the ascending order of occupancy times of space-time conflicting occupancies and returns the set. Note that $\vert \mathcal{C} \vert \le n$ and $n \ll N_m$ in general. Hence, this sorting operation can be done with $\mathcal{O}(n log_2 N_m)$ computational complexity. If we consider all these calculation steps in the \verb+getCV()+ function, then one can see that the overall computational complexity for space-time conflict checking steps in \verb+getCV()+ is $\mathcal{O}(n N_m^3)$

\emph{DTOT adjustment for collision avoidance}: 
Once the set $\mathcal{C}$ is returned from the function \verb+getCV()+, the DICA algorithm updates the vehicle $v^i$'s DTOT to avoid space-time conflict with DTOTs of the vehicles in the set $\mathcal{C}$. In line 7 (Algorithm \ref{code:main}), it is shown that the first vehicle $v^j$ in the set $\mathcal{C}$ is considered for updating the vehicle $v^i$'s DTOT. As described in Section \ref{sec:dtot}, our update strategy to avoid space-time conflict is to make the vehicle $v^i$ enter the intersection area a little bit late so as to give enough time for vehicle $v^j$ to cross the intersection safely. For this, the DICA algorithm first needs to compute the delay time needed to avoid the space-time conflict with the vehicle $v^j$. Since the occupancy time interval $I(O^j_k)$ for the vehicle $v^j$'s earliest space-time conflicting occupancy has already been determined from the function \verb+getCV()+, it is easy to calculate this delay time in this update process. Once the delay time is determined, then the remaining step is simply to change the occupancy times of all the occupancies in the vehicle $v^i$'s DTOT to be delayed and this results in $\mathcal{O}(N_m)$ computational complexity.

As described above, the number of vehicles in the set $\mathcal{S}$ is $n$ when the function \verb+getCV()+ is called for the first time in line 5 (Algorithm \ref{code:main}). Then, within the \verb+while+ loop, the function \verb+updateDTOT()+ adjusts the vehicle $v^i$'s DTOT to avoid collision with the first vehicle in the set $\mathcal{C}$ and this step reduces the number of vehicles in the set $\mathcal{C}$ that can potentially collide with the vehicle $v^i$ at least by one. Thus, in the worst case, the number of vehicles in the set $\mathcal{C}$ returned by the second call of $\verb+getCV()+$ within the \verb+while+ loop is $(n-1)$. If we assume the worst case for every following iterations within the \verb+while+ loop until the set $\mathcal{C}$ becomes empty, then it is easy to see that the functions \verb+getCV()+ and \verb+updateDTOT()+ are called $n$ times within the \verb+while+ loop. This implies that, since the computational complexity of the function \verb+updateDTOT()+ is significantly lower than that of the function \verb+getCV()+, the overall computational complexity of the \verb+while+ loop can be considered as $\mathcal{O}(n^2 N_m^3)$. 

Note that the maximum number of occupancies $N_m$ depends on both the time that it takes for a vehicle to cross the intersection and the discrete time step used to construct the DTOT by ICA. If we let $h$ be the discrete time step used by ICA and $T_m$ be the time it takes for a vehicle to completely cross an intersection when the vehicle starts from rest and accelerates to cross the intersection as quickly as possible, then we have $\bar{N}_m := T_m/h$ as an upper bound for $N_m$. Note that $T_m$ depends on the length of an intersection crossing route that a vehicle takes to cross an intersection. If we let $L_m$ be the maximum length out of all intersection crossing routes for an intersection, then $\bar{N}_m$ can be expressed in terms of $L_m$ instead of $T_m$. Specifically, if $L_m$ is long enough so that a vehicle can reach its maximum allowed speed $v_m$ within an intersection before it completely crosses the intersection, then it can be shown that $\bar{N}_m = {(2a_m L_m + v_m^2) / (2a_m v_m h)}$ where $a_m$ is the maximum acceleration rate of a vehicle. On the other hand, if $L_m$ is not long enough for a vehicle to reach $v_m$ while crossing an intersection, then it is also relatively straightforward to show that $\bar{N}_m = (\sqrt{2 L_m / a_m})/h$. (These two different cases are illustrated in Figure \ref{fig:complexity}.)  If we fix values for $h, v_m$, and $a_m$, then one can see that $\bar{N}_m$ for the former case is proportional to $L_m$ while, for the latter case, $\bar{N}_m$ is proportional to the square root of $L_m$. Hence, if we substitute $L_m$ for $N_m$ in the computational complexity $\mathcal{O}(n^2 N_m^3)$ that we derived above, then we finally have $\mathcal{O}(n^2 L_m^3)$ as the overall computational complexity of DICA.
\begin{figure}
	\centering
	\includegraphics[width=2.5in]{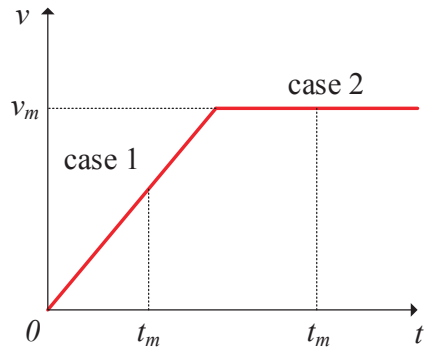}
	\caption[Two different cases for shortest intersection crossing time ($T_m$) calculation]{Two different cases for shortest intersection crossing time ($T_m$) calculation. (Case 1 is the situation when $L_m$ is too short to reach $v_m$ and case 2 is the situation when $L_m$ is long enough to reach $v_m$ while a vehicle is crossing an intersection.)}
	\label{fig:complexity}
\end{figure}
\end{proof}

\section{Algorithm Improvements} \label{sec:improv}
According to the computational complexity analysis result described in the previous section, it is true that the original DICA algorithm that is shown in Algorithms \ref{code:main} and \ref{code:GetCV} is somewhat conservative in terms of computational cost to be used in practice. 
In this section, we present several approaches that can be used to improve the overall computational complexity of the algorithm.

\subsection{Reduced Number of Vehicles for Space-Time Conflict Check} \label{sec:imp1}
As shown in Algorithm \ref{code:GetCV}, all confirmed vehicles in the set $\mathcal{S}$ are examined to obtain the set of space-time conflicting vehicles $\mathcal{C}$ for a new unconfirmed head vehicle $v^i$. However, we see that this computation process can be improved by excluding vehicles that cannot be in space-time conflict with the vehicle $v^i$ under any circumstances from the set $\mathcal{S}$. For example, a confirmed vehicle $v^j \in \mathcal{S}$ who has an intersection crossing time interval that is not overlapping with the vehicle $v^i$'s intersection crossing time interval can be excluded. Note that the intersection crossing time interval of a confirmed vehicle can be easily determined by the lower bound of the occupancy time $\tau_{lb}(O_{first})$ of the vehicle's first occupancy $O_{first}$ and the upper bound of the occupancy time $\tau_{ub}(O_{last})$ of the vehicle's last occupancy $O_{last}$ in the vehicle's confirmed DTOT. In addition to these vehicles, vehicles in the set $\mathcal{S}$ whose intersection crossing routes are compatible with that of vehicle $v^i$ can also be excluded. 
Hence, if we let $\mathcal{S}^*$ be the subset of all confirmed vehicles in set $\mathcal{S}$ that can be obtained after excluding all above mentioned vehicles in determining the set $\mathcal{C}$, then the resulting computational complexity for the space-time conflict checking in function \verb+getCV()+ becomes $\mathcal{O}(\alpha_1 n N_m^3)$ where $\alpha_1 := \tilde{n}/n$, $\tilde{n} = \vert \mathcal{S}^*\vert$, $n = \vert \mathcal{S} \vert$, and $N_m$ is the maximum number of occupancies of all vehicles that are in the set $\mathcal{S}$ and also the vehicle that is currently under consideration for confirmation. (See the proof of Proposition \ref{prop:dica} for the precise definition of $N_m$.)

\subsection{Efficient Space Conflict Check} \label{sec:imp2}
Note that, for any two vehicles coming from different directions, they can collide with each other only within some parts of their intersection crossing routes. Thus, not all occupancies of a vehicle's DTOT needs to be checked for space conflict with another vehicle's DTOT. For example, the two vehicles $v^i$ and $v^j$ in Figure \ref{fig:example} have very short ranges of intersection crossing routes that are space conflicting with each other. Thus, the occupancies to be checked can be reduced to {\{$O^i_2, O^i_3$\} and \{$O^j_5, O^j_6$\} from their entire DTOTs. 
Since the number of occupancies in a DTOT is very large in general, this can improve computational speed considerably. 
Note that, since the intersection crossing routes are fixed for a specific intersection, we can predetermine these space conflicting short ranges offline only one time for all pairs of incompatible intersection crossing routes. Hence, this extra preparation process does not incur an additional computational cost during the online operation of DICA.
If we use DTOT$^*$ to denote the subset of the original DTOT for a vehicle that can be obtained from this approach, then the computational complexity of the function \verb+getCV()+ in Algorithm \ref{code:GetCV} can be expressed as $\mathcal{O}(\alpha_2^3 n N_m^3)$ where $\alpha_2 := \tilde{N}_m / N_m$ and $\tilde{N}_m$ is the maximum number of occupancies of all vehicles that are in the set $\mathcal{S}^*$ and the vehicle that is currently under consideration for confirmation.

\subsection{Approximate Occupancy Time Interval Calculation} \label{sec:imp3}
 As explained in Section 3, ICA checks if an occupancy of a vehicle is conflicting in time with another vehicle's occupancy using occupancy time intervals that can be obtained from each vehicle's DTOT. However, the way to obtain an occupancy time interval presented in the proof of Proposition \ref{prop:dica} is somewhat naive in the sense of computational complexity. In fact, as analyzed in the proof, such an exhaustive search involves computational complexity of $\mathcal{O}(N_m)$. To simplify this computation process, we propose to estimate the occupancy time interval for a certain occupancy based on the vehicle's speed, length, and acceleration rate instead of performing the exhaustive search. To clarify this idea, let us consider an example. For simplicity of explanation, we consider a case when a vehicle is moving in a straight line as shown in Figure \ref{fig:oti}. Let $O^i_k$ be the occupancy for which the DICA algorithm needs to determine the occupancy time interval $I(O^i_k) = [\tau_{lb}(O^i_k), \tau_{ub}(O^i_k)]$, $L(v^i)$ be the vehicle length of the vehicle $v^i$, $h$ be the sampling time interval, $x_k$ be the center position of the $O^i_k$ along the straight line. Then the algorithm first estimates the vehicle's speed and acceleration rate around the occupancy $O^i_k$ from $x_k$, $x_{k-1}$, $x_{k+1}$, and $h$. Occupancies at $x_{k-1}, x_{k+1}$ are very close to the occupancy $O^i_k$ and are not shown in Figure \ref{fig:oti} for simplicity. Specifically, if we let $V_{k^-}(v^i)$ and $V_{k^+}(v^i)$ be the speed of the vehicle $v^i$ from $O^i_{k-1}$ to $O^i_k$ and from $O^i_{k}$ to $O^i_{k+1}$ respectively, then these speeds can be approximated as follows:
\begin{equation} \nonumber
	V_{k^-}(v^i) \approx \frac{x_k - x_{k-1}}{h} , \quad V_{k^+}(v^i) \approx \frac{x_{k+1} - x_{k}}{h}
\end{equation}
From these speeds, we now approximate the acceleration rate of the vehicle as follows:
\begin{equation} \nonumber
	A_k(v^i) \approx \frac{V_{k^+}(v^i) - V_{k^-}(v^i)}{h}
\end{equation}
where $A_k(v^i)$ denotes the acceleration of the vehicle $v^i$ at the occupancy $O^i_k$. If we take the average of the speeds around $O^i_k$, then we can also approximate $V_k(v^i)$ which is the speed of the vehicle $v^i$ at $O^i_k$. Note that since the length of a vehicle $L(v^i)$ is just a few meters in general, the actual motion of the vehicle $v^i$ within the occupancy $O^i_k$ can be approximated fairly accurately by $V_k(v^i)$ and $A_k(v^i)$. 

Now, since it is a straightforward process to estimate $\tau_{lb}(O^i_k)$ and $\tau_{ub}(O^i_k)$ from $L(v^i)$,  $V_k(v^i)$, and $A_k(v^i)$, we omit the details of these calculations in this paper. For the case when the vehicle is moving on a curved path, we can still use the same method to approximate $V_k(v^i)$ and $A_k(v^i)$. But, in this case, we may need to add a short extra distance to the $L(v^i)$ to estimate $\tau_{lb}(O^i_k)$ and $\tau_{ub}(O^i_k)$ more accurately. Such an extra distance can be simply determined by the curvature of the path that is represented by the DTOT of a vehicle. 
Finally, if we apply this approximation method for an occupancy time interval calculation in the \verb+getOTI()+ function, then the computational complexity of the function \verb+getCV()+ improves from $\mathcal{O} (n^2 N_m^3)$ to $\mathcal{O}(n^2 N_m^2)$.  
\begin{figure}
	\centering
	\includegraphics[width=3in]{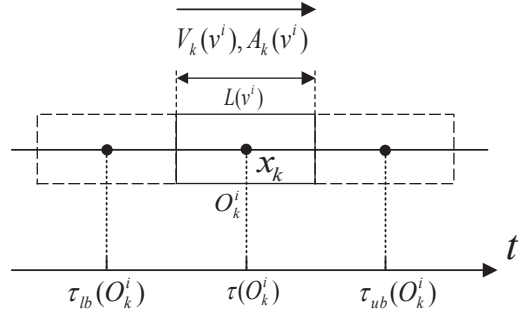}
	\caption{Approximate occupancy time interval calculation for a vehicle with the through route}
	\label{fig:oti}
\end{figure}

\subsection{Efficient Occupancies Comparison} \label{sec:imp4}
In addition to all the techniques described above, the overall computational complexity of the Algorithm \ref{code:main} can be improved further if we employ an efficient searching method such as the bisection method in the process of time-conflict checking between two DTOT$^*$s.
If we employ this bisection approach for time-conflict checking as shown in Algorithm \ref{code:Enhanced_GetCV}, then the computational complexity of the function \verb+getCV()+ can be improved significantly from $\mathcal{O}(n^2 N_m^3)$ to $\mathcal{O}( n^2 N_m^2 \log_2 N_m)$. 


All of the improvement techniques discussed in this section are incorporated into the function \verb+getCV()+ to improve the overall computational complexity of the space-time conflict checking process. Algorithm \ref{code:Enhanced_GetCV} shows this modified \verb+getCV()+ function which is now called \verb+enhanced_getCV()+. In Algorithm \ref{code:Enhanced_GetCV}, $\mathcal{S}^*$ represents the set of already confirmed vehicles that is obtained from the process in Section \ref{sec:imp1} and DTOT$^*$ represents the subset of original DTOT for a vehicle that can be obtained from the approach in Section \ref{sec:imp2}. The function \verb+getOTI()+ within the \verb+while+ loop is now replaced by the new function \verb+getEstOTI()+ that calculates the occupancy time interval approximately as described in Section \ref{sec:imp3}. Lastly, the approach for efficient time conflict checking that is presented in Section \ref{sec:imp4} is implemented throughout the \verb+while+ loop of the DICA algorithm.

\begin{algorithm}[!htb]   
	\caption{enhanced\_getCV$(\mathcal{S}^*, DTOT(v^i))$}
		\begin{algorithmic}[1]
		\STATE $\mathcal{C} = \emptyset$
		\FOR {$v^j$ in $\mathcal{S}^*$}
		\FOR {$O^j_{k_j}$ in $DTOT^*(v^j)$} \label{O}
		\IF {$v^j$ not in $\mathcal{C}$}
		\STATE $high = \vert DTOT^*(v^i) \vert - 1$
		\STATE $low = 0$
		\WHILE {$low \ne high$}
		\STATE $middle = {(high + low)/ 2}$
		
		\STATE Call getEstOTI($O^j_{k_j})$ $\rightarrow I(O^j_{k_j}) := [\tau_{lb}({O}^j_{k_j}), \tau_{ub}({O}^j_{k_j})]$
		\STATE Call getEstOTI($O^i_{middle}$) $\rightarrow I(O^i_{middle}) := [\tau_{lb}({O}^i_{middle}), \tau_{ub}({O}^i_{middle})]$
		\IF {$I({O}^j_{k_j}) \cap I({O}^i_{middle}) \neq \emptyset$}
		\STATE Assign $\tau_{lb}({O}^j_{k_j}) \rightarrow v^j.firstTimeAtCollision$
		\STATE Push $v^j$ into $\mathcal{C}$
		\STATE Sort $\mathcal{C}$ in ascending order of \emph{firstTimeAtCollision}
		\ELSIF{$\tau({O}^j_{k_j}) > \tau({O}^i_{middle})$}
		\STATE $low = middle$
		\ELSIF{$\tau({O}^j_{k_j}) < \tau({O}^i_{middle})$}
		\STATE $high = middle$
		\ENDIF

		\ENDWHILE
		\ENDIF
		\ENDFOR
		\ENDFOR

	\end{algorithmic}
	
	\label{code:Enhanced_GetCV}
\end{algorithm}

\begin{prop} \label{prop:edica}
Enhanced DICA has $\mathcal{O}(\alpha n^2 L_m  \log_2 L_m)$ computational complexity where $\alpha := \alpha_1^2 \alpha_2 \ll 1$, $n$ is the number of vehicles already confirmed to cross an intersection, and $L_m$ is the maximum length of intersection crossing routes in an intersection.
\end{prop} 
\begin{proof}
First, note that the only part in Algorithm \ref{code:main} that is affected by this proposed enhancement is that the number of confirmed vehicles to be considered for a space-time conflict check is reduced from $n = \vert \mathcal{S} \vert$ to $\tilde{n} = \vert \mathcal{S}^* \vert$ where $\tilde{n} = \alpha_1 n$ and $\alpha_1 \in (0, 1]$. Thus, in Algorithm \ref{code:main}, the functions \verb+enhanced_getCV()+ and \verb+updateDTOT()+ are now called $\alpha_1 n$ times. 
Next, we also note that, since nothing is changed due to this improvement in the \verb+updateDTOT()+ function whose computational complexity is already significantly lower than that of the function \verb+getCV()+, it suffices to analyze the computational complexity of the function \verb+enhanced_getCV()+ presented in Algorithm \ref{code:Enhanced_GetCV} for the overall computational complexity of the enhanced DICA. 

Now, as one can see in Algorithm \ref{code:Enhanced_GetCV}, the entire block within the outer \verb+for+ loop is executed for $\alpha_1 n$ times since the number of confirmed vehicles to be checked for a space-time conflict with the vehicle $v^i$ is reduced from $n$ to $\alpha_1 n$ due to the approach discussed in Section \ref{sec:imp1}. Then, within the \verb+for+ loop, for each vehicle $v^j$ in the set $\mathcal{S}^*$, occupancies from each vehicle's DTOT are evaluated for space and time conflict which \emph{typically} requires $N_m^2$ times occupancy comparison operation where $N_m$ is the maximum number of occupancies in a vehicle's DTOT. However, in the \verb+enhanced_getCV()+ function, we first note that the maximum number of occupancies for each vehicle's DTOT to be tested for space-time conflict is reduced from $N_m$ to $\tilde{N}_m$ where $\tilde{N}_m = \alpha_2 N_m$ and $\alpha_2 \in (0, 1]$ due to the approach presented in Section \ref{sec:imp2}. Another important improvement is that the computational complexity for the occupancy time interval calculation is improved from $\mathcal{O}(N_m)$ to $\mathcal{O}(1)$ within another enhanced function \verb+getEstOTI()+ as discussed in Section \ref{sec:imp3}. Therefore, the overall computational complexity of the outer \verb+for+ loop can be estimated as $\mathcal{O}(\alpha_1 \alpha_2^2 n N_m^2)$. However, note that this is the case when we use the same occupancies comparison method as used in the original \verb+getCV()+ function. As shown in Algorithm \ref{code:Enhanced_GetCV}, the process of occupancies comparison is now performed based on the bisection search method. Roughly speaking, for given $n$ and $N_m$, this efficient search method improves the overall computational complexity of the function from $\mathcal{O}(n N_m^2)$ to $\mathcal{O}(n N_m \log_2 N_m)$ as discussed in Section \ref{sec:imp4}. If we combine this and others discussed above for the overall computational complexity of the \verb+enhanced_getCV()+ function, then we have $\mathcal{O}(\alpha_1 \alpha_2 n N_m \log_2 N_m)$. Recall that the \verb+enhanced_getCV()+ function is called at $\alpha_1 n$ times in the main \verb+while+ loop as discussed above, we have $\mathcal{O}(\alpha_1^2 \alpha_2 n^2 N_m \log_2 N_m)$ as the overall computational complexity of DICA.

As we have analyzed already in the proof of Proposition \ref{prop:dica}, $N_m$ is linearly proportional to the maximum length of intersection crossing routes $L_m$. Hence, if we substitute $L_m$ for $N_m$, then we finally have $\mathcal{O}(\alpha n^2 L_m \log_2 L_m)$ as the overall computational complexity of enhanced DICA where $\alpha := \alpha_1^2 \alpha_2 \ll 1$. 
\end{proof}

\section{Simulation} \label{sec:sim}

In this section, we present simulation results that demonstrate the improved performance of the enhanced DICA over the original algorithm. The performance of the enhanced algorithm is also compared with that of an optimized traffic light intersection control.
\subsection{Simulation Setup}
To evaluate the performance of the original DICA and the enhanced DICA, we implemented both algorithms in a microscopic road traffic simulation software, called the Simulation of Urban MObility (SUMO) \cite{krajzewicz2012recent}, and performed extensive intersection traffic simulations. In our simulation, the simulated situation is an intersection crossing traffic on a typical isolated four way intersection with three incoming lanes, one of which is a dedicated lane for left-turning vehicles, and two outgoing lanes on each road as shown in Figure \ref{fig:screenshot}. We set $70 km/h$ as the maximum allowed speed $v_m$ for all incoming vehicles. To make the simulation more realistic, we let vehicles approach an intersection with different speeds when they enter into the communication region of the intersection. Specifically, when a new vehicle is spawned outside of the communication region, we assign the initial speed of the vehicle randomly within the range from 40\% to 100\% of the maximum allowed speed $v_m$. Thus, a vehicle keeps this random initial speed until it enters the communication region and then it either follows another vehicle or is confirmed by ICA with a feasible DTOT.
The maximum acceleration ($a_{max}$) and deceleration ($a_{min}$) rates for vehicles that are used in simulations are $2 m/s^2$ and $4.5 m/s^2$, respectively. The size of a vehicle used in simulations is $5$ meters long and $1.8$ meters wide. Since, in some cases, a vehicle may need to stop just before the entrance line of the intersection region to avoid collisions with other vehicles, the distance from the entrance line of the communication region to the entrance line of the intersection region should be long enough so that a vehicle can stop from its maximum speed $v_m$. Thus, from the value used for $v_m = 70 km/h$ and the maximum deceleration rate $a_{min} = 4.5 m/s^2$, we need at least $v_m^2/(2 a_{min}) \approx 42.03 m$. So, we use $50 m$ for the distance from the entrance line of the communication region to the entrance line of the intersection region.
The time step that is used in simulation is $0.05$ seconds. In most cases, a simulation terminates when the simulation time reaches 10 minutes. 

In our simulations, vehicles are spawned according to several random variables in order to generate various traffic volumes as well as traffic patterns. Specifically, a Bernoulli random variable $X_{V}$ is used to spawn a vehicle at each simulation time step on each incoming road. In particular, a vehicle is spawned if $X_V = 1$ and not spawned if $X_V = 0$ with $Pr(X_V = 1) = p_V$ where $Pr(E)$ is the probability of an event $E$ and $0 \le p_V \le 1$ is the probability that a vehicle is spawned. Thus, by adjusting $p_V$, we can generate various traffic volumes.  
In addition to this traffic volume variation, we also assign the probabilities for each vehicle to have different routes. We use a three-states random variable $X_P$ to assign an intersection route to a vehicle probabilistically when the vehicle is generated. Three states of the random variable $X_P$ are $\{ Left, Straight, Right\}$ with $Pr(X_P = Left)  = p_L$, $Pr(X_P = Straight)  = p_S$, and $Pr(X_P = Right)  = p_R$ being the probability of turning left, going straight, and turning right to cross an intersection, respectively. As shown in Table \ref{rv_simulation}, we set $p_L = 0.2$, $p_S = 0.6$, and $p_R = 0.2$ for all traffic volume cases to generate $20\%$ of all incoming vehicles for left turing, $60\%$ for going straight, and the other $20\%$ for right turning. 
To create variations on the traffic pattern with this random variable $X_P$, we generate a random number at each simulation time step from the range $[0, 1]$ with uniform distribution. However, a random number generated in programming languages is not completely random in general due to its dependency on the seed number used for random number generation. For example, if we generate a random number with a fixed random seed number, the random number generated at simulation time step $t$ will be the same whenever we run a simulation. Thus, we use several different random seed numbers to generate different traffic patterns. 
Table \ref{rv_simulation} summarizes the parameters used for various traffic volumes and patterns that were used in many of our simulations. As shown in the table, we use three random seeds to generate three different intersection traffic patterns for each traffic volume. Thus, to obtain simulation data for each traffic volume, we run three simulations with different traffic patterns for each simulation and then use the averages of these simulation results as the result for each traffic volume case. 
The intersection crossing traffics generated in most of our simulations are balanced traffics in the sense that the number of vehicles generated in each incoming road are about the same. However, for a simulation to show the starvation free property of the proposed DICA algorithm, the intersection traffic is purposely designed to be unbalanced where the number of vehicles for minor approaching roads is roughly 30\% of the vehicles in major roads.

%
%
%
 \begin{table}[tb]\scriptsize	
 	\centering
 	\vspace{-0.5cm}
 	\caption{Parameters used for various traffic volumes and patterns. ($*$ Expected number of vehicles per 10 minutes.)}
 	\smallskip
 	\label{rv_simulation}
 	\begin{tabular}{c | c }
 		\hline
 		Parameter & Value \\
 		\hline 
 		Traffic volumes$*$ & 100 / 200 / 300 / 400 / 500 \\
 		\hline
 		$p_V$ & 0.03 / 0.06 / 0.08 / 0.11 / 0.14 \\ 
 		\hline
 		$p_L$  & 0.20  \\ 
 		\hline
 		$p_S$  & 0.60   \\ 
 		\hline
 		$p_R$ & 0.20  \\ 
 		\hline
 		Random seeds  & 12 / 21 / 66   \\ 
 		\hline	
 		
 	\end{tabular} 
 \end{table}

In the following discussion on our simulation results, \emph{simulation time} means the simulated time used in simulation program and \emph{computation time}, which will be discussed later in Section \ref{sec:sim:time}, means the actual elapsed time that it takes for a computer to run a simulation. 
Also, in Section \ref{sec:sim:perf}, the traffic control performance of the enhanced DICA is compared with that of a traffic light algorithm with fixed cycles. To have a comparable traffic light program, we computed the optimal signal cycles for different traffic volume cases by using the exponential cycle length model $C_0 = 1.5 L e^{1.8 Y}$ from \cite{cheng2005development}. In the model, $L$ represents the total lost time within the cycle. The lost time for each phase is assumed to be 4 seconds \cite{HCM2000}. Thus, $L = 4 \times 4s = 16s$. $Y$ is the sum of critical phase flow ratios. The duration for the yellow light of each phase is 3 seconds. 

All simulations were run on a Core i7 computer with 3.40GHz, 8GB RAM and Windows system. The interface programs with SUMO were coded in Python.

\begin{figure}[!h]
	\centering
	\includegraphics[width=3.3in]{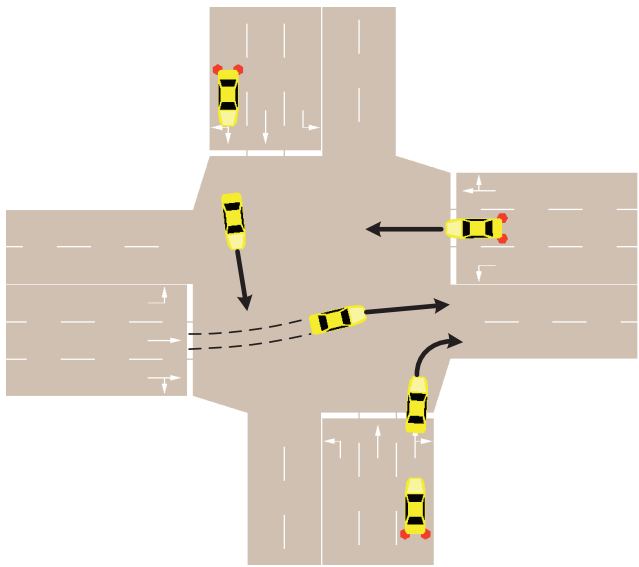}
	\caption[A screenshot of simulation]{A screenshot of simulation which illustrates a situation when vehicles with conflicting routes cross an intersection simultaneously.}
	\label{fig:screenshot}
\end{figure}

\subsection{Simulation Results}

Computation times and performances of three different traffic patterns for all five volume cases are recorded from simulations. 
Figure \ref{fig:screenshot} shows a screenshot of simulation in SUMO when vehicles of different routes are crossing an intersection simultaneously without occurrence of collisions. In this situation, the through vehicle from North goes inside the intersection shortly after the vehicle from West to East clears the conflicting space and the through vehicle from East starts to enter the intersection although the route of the vehicle is conflicting with that of the vehicle from North. Vehicles whose DTOTs are not conflicting with these two can pass the intersection at the same time. For example, the right-turning vehicle from South also crosses the intersection at the same time as the other vehicles in Figure \ref{fig:screenshot}. 


 \subsubsection{Computation Time} \label{sec:sim:time}
Figure \ref{fig:SimulationTime} (a) compares the computation times of the original DICA, the enhanced DICA, and the optimized traffic light algorithm. Figure \ref{fig:SimulationTime} (b) shows how much computational improvement was made through each computational improvement technique discussed in Sections \ref{sec:improv}-A, \ref{sec:improv}-B, \ref{sec:improv}-C, and \ref{sec:improv}-D. Note that since the computational improvement technique in Section\ref{sec:improv}-D is implemented based on the computational improvement technique in Section \ref{sec:improv}-B, we had to combine techniques from both Sections \ref{sec:improv}-D and \ref{sec:improv}-B to show the improvement due to the technique in Section \ref{sec:improv}-D indirectly. 
Here, we show the computation times comparison for only one traffic volume case with 300 vehicles per 10 minutes since the trends for other volume cases are similar. 
The vertical axis in Figure \ref{fig:SimulationTime} is the computation time in hour unit which is represented in logarithmic scale. As shown in Figure \ref{fig:SimulationTime}, the enhanced DICA that implements all improvements discussed in Section \ref{sec:improv} takes significantly less computation time, i.e. only $0.4\%$ computation time of the original algorithm. 
When we apply each computational improvement technique individually, our result shows that it takes about $11\%$ of the computation time of the original DICA with the technique in Section \ref{sec:improv}-A, $59\%$ with the technique in Section \ref{sec:improv}-B, $13\%$ with the technique in Section \ref{sec:improv}-C, and $6\%$ with techniques in Sections \ref{sec:improv}-B and Section \ref{sec:improv}-D together. If we combine all of these individual improvement altogether to estimate the collective improvement, then we have about $0.45\%$ computation time of the original DICA which is similar to the computation time result with the enhanced DICA in which all these techniques are implemented. 
 
\begin{figure}[!t]
	\centering
	\includegraphics[width=6.6in]{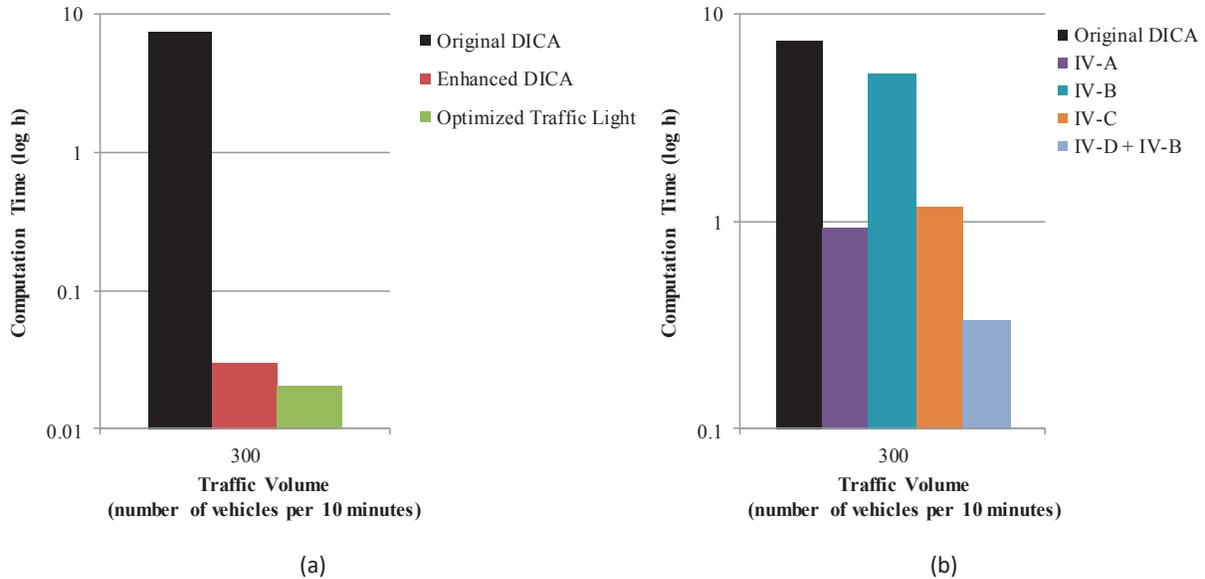}
	\caption[Computation times comparison]{Computation times comparison for traffic volume with 300 vehicles per 10 minutes: (a) original DICA with different algorithms, (b) original DICA with different improvement techniques. (The symbol IV-x represents the improvement technique in Section IV-x where x = $\{$ A, B, C, D $\}$.)}
	\label{fig:SimulationTime}
\end{figure}
 
Table \ref{table_time} compares the computation times between the enhanced DICA and the optimized traffic light algorithm for all five traffic volume cases. 
From the results shown in the table, we note that the computation time for the optimized traffic light algorithm gradually increases as the traffic volume increases. 
However, since the optimized traffic light algorithm has $\mathcal{O}(1)$ computational complexity, its computation time cannot be affected by the number of vehicles around an intersection. Thus, roughly speaking, one can say that the computation time of the optimized traffic light for a particular traffic volume case is in fact the time required for the simulation software SUMO to run a simulation with the number of vehicles for that particular traffic volume case. Therefore, the actual computation time of the enhanced DICA for a particular traffic volume case can be roughly approximated by subtracting the computation time of the optimized traffic light for the case from the computation time of the enhanced DICA presented in the table. For example, for the traffic volume with 500 vehicles, the actual computation time for the enhanced DICA can be approximated as $0.031 (= 0.058 - 0.027)$ hours which is 1.86 minutes.    
Note that this 1.86 minutes is the computation time taken by the algorithm to handle 500 vehicles. Thus this in turn implies that it takes only 0.2232 seconds to handle each vehicle. 
An exception to this approximation is the case with 100 vehicles traffic volume case where the computation time for optimized traffic light takes longer time than that of the enhanced DICA. The reason for this result can be understood by considering the fact that, in such a low traffic volume situation, the average number of vehicles to be simulated by SUMO at each simulation time step is smaller in the enhanced DICA case since vehicles are crossing an intersection much faster without waiting at an intersection under the enhanced DICA than the optimized traffic light as shown in Section \ref{sec:sim:perf}.

 \begin{table}[tb]\scriptsize	
 	\centering
 	\vspace{-0.5cm}
 	\caption{Computation time comparison between enhanced DICA and optimized traffic light}
 	\smallskip
 	\label{table_time}
 	\begin{tabular}{cccccc}
 		\hline 
 		Traffic vsolume &  \multirow{2}{*}{100} &  \multirow{2}{*}{200}& \multirow{2}{*}{300} & \multirow{2}{*}{400} & \multirow{2}{*}{500} \\ 
 		(Number of vehicles per 10 minutes) &  & &  &  &  \\
 		
 		\hline 
 		
 		Optimized Traffic light (h) & 0.014 & 0.017 & 0.020 & 0.024 & 0.027 \\ 
 		Enhanced DICA (h) & 0.011 & 0.024 & 0.026 & 0.042 & 0.058 \\ 
 		
 		\hline 
 		
 	\end{tabular} 
 \end{table}

 \subsubsection{Liveness and Safety}
Although we have theoretically showed the liveness of DICA, it is better to have simulation results that support the theory. Since the simulation in this section is only a verification, we run a simulation with 10, 000 vehicles instead of giving a restriction on the simulation time. The simulation ends after all 10, 000 vehicles have exited the simulation scene. We recorded the number of vehicles that are waiting to cross the intersection at each simulation time step and plot the number profile in Figure \ref{fig:waitToCross}. As shown in the figure, the number of vehicles drops to zero in almost a linear way within a finite time which demonstrates that every vehicle was able to cross the intersection eventually which proves the proposition 1 in Section III-A. 
We also performed a set of simulations for the case of unbalanced traffic situation where the number of vehicles on minor roads is only 30\% of that of major roads to demonstrate the fairness of DICA.
To show the fairness of the algorithm, we recorded the average trip times for major roads and minor roads respectively for every traffic volumes. 
As shown in Table \ref{table}, one can find that the average trip time of the minor roads is about the same as that of the major roads. This shows that there is not a case that some vehicles cannot get confirmation or will experience a very long time to be confirmed which demonstrates the proposition 2 in Section III-A. 

\begin{figure}[!t]
	\centering
	\includegraphics[width=3.3in]{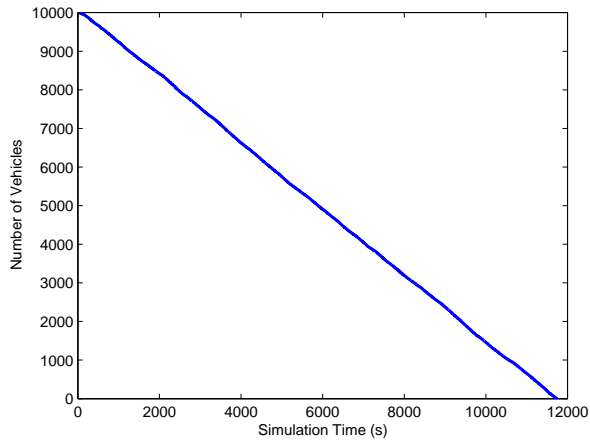}
	\caption{The number of vehicles which wait to cross the intersection over time.}
	\label{fig:waitToCross}
\end{figure}

\begin{table}[h]\scriptsize	
	\centering
	\vspace{-0.5cm}
	\caption{Average trip time comparison between major roads and minor roads in an unbalanced traffic}
	\smallskip
	\label{table}
	\begin{tabular}{cccccc}
		\hline 
		Traffic volume &  \multirow{2}{*}{100} &  \multirow{2}{*}{200}& \multirow{2}{*}{300} & \multirow{2}{*}{400} & \multirow{2}{*}{500} \\ 
		(Number of vehicles per 10 minutes) &  & &  &  &  \\
		
		\hline 

		Average trip time on major roads (s) & 6.17 & 6.60 & 7.38 & 8.15 & 10.15 \\ 
		Average trip time on minor roads (s) & 6.21 & 6.57 & 7.38 & 7.90 & 9.63 \\ 
		
		\hline 
		
	\end{tabular} 
\end{table}

To validate the safety property (i.e., collision freeness) of DICA through simulation, we computed the inter-vehicle distance between every pair of vehicles within an intersection at every second in simulation time. Since each vehicle is represented as a polygon, a 5 $m$ long and 1.8 $m$ wide rectangle more precisely, we obtained this data based on an algorithm of the shortest distance calculation between two polygons. 
A histogram of the recorded inter-vehicle distances is shown in Figure \ref{fig:minDistance}. 
Clearly, the inter-vehicle distance must be less than or equal to zero if two vehicles are in a collision and must be positive otherwise. 
As one can see from the figure, there is no instance observed throughout the entire simulation with less than $1 m$ inter-vehicle distance, which is a clear indication that there is no collision inside the intersection. Note that Figure \ref{fig:minDistance} is demonstrating the safety of the DICA algorithm, the safety problem that vehicles cannot follow confirmed DTOT correctly pertaining to the robustness of DICA will be studied in our future work.
 \begin{figure}[!t]
	\centering
	\includegraphics[width=3.5in]{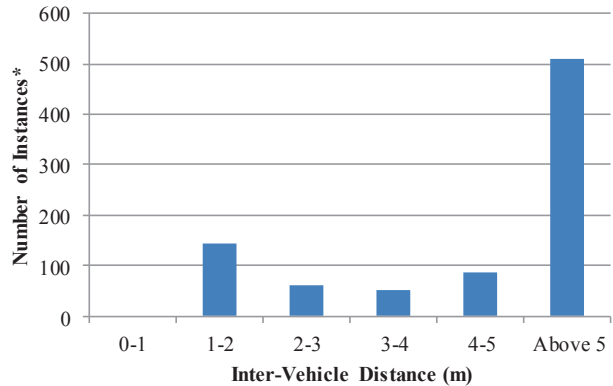}
	\caption[Histogram of the inter-vehicle distance within the intersection.]{Histogram of the inter-vehicle distance within the intersection. (* An instance means the situation when a pair of vehicles are separated by the calculated inter-vehicle distance.)}
	\label{fig:minDistance}
\end{figure}

\subsubsection{Control Performance} \label{sec:sim:perf}
The overall traffic control performance of the enhanced DICA is also evaluated and compared with that of the optimized traffic light algorithm based on the following performance measures. For each vehicle, we recorded the \emph{trip time} that is the time taken for a vehicle from the moment when it enters into the communication region of an intersection until the vehicle completely crosses the intersection region. From the recorded trip time data for all crossed vehicles, we calculated several related statistic information which are the \emph{average trip time} and the \emph{standard deviation of trip time}. Besides these trip time related performance measures, we also calculated the percentage of all crossed vehicles against the total number of generated vehicles, which we call the \emph{throughput}.
However, note that neither the average trip time nor the throughput alone is sufficient to correctly evaluate the performance of an algorithm. In fact, both of these measures should be considered together to correctly compare and evaluate the performances of different intersection traffic control algorithms. For this reason, we calculated the ratio of average trip time to throughput, which we call the \emph{effective average trip time}, and believe that this could show performance of an algorithm better.
Comparison of the performance between the enhanced DICA and the optimized traffic light control algorithm are shown in Figure \ref{fig:performance}. 
From this result, we can see that, since the throughputs of the two algorithms are always similar, the profiles of average trip time and effective average trip time also show similar trends. The enhanced DICA always performs better than optimized traffic light for the first four traffic volume cases. In the case of the traffic volume with 500 vehicles, the average trip time performance of the enhanced DICA becomes closer to that of optimized traffic light. Also, the enhanced DICA has a bit larger standard deviation of trip time than the optimized traffic light. In short, the enhanced DICA performs much better than the optimized traffic light from low to medium traffic volume cases while its performance becomes worse and closer to the performance of the optimized traffic light for heavy traffic volumes.  

\begin{figure}[!t]
	\centering
	\includegraphics[width=6.6in]{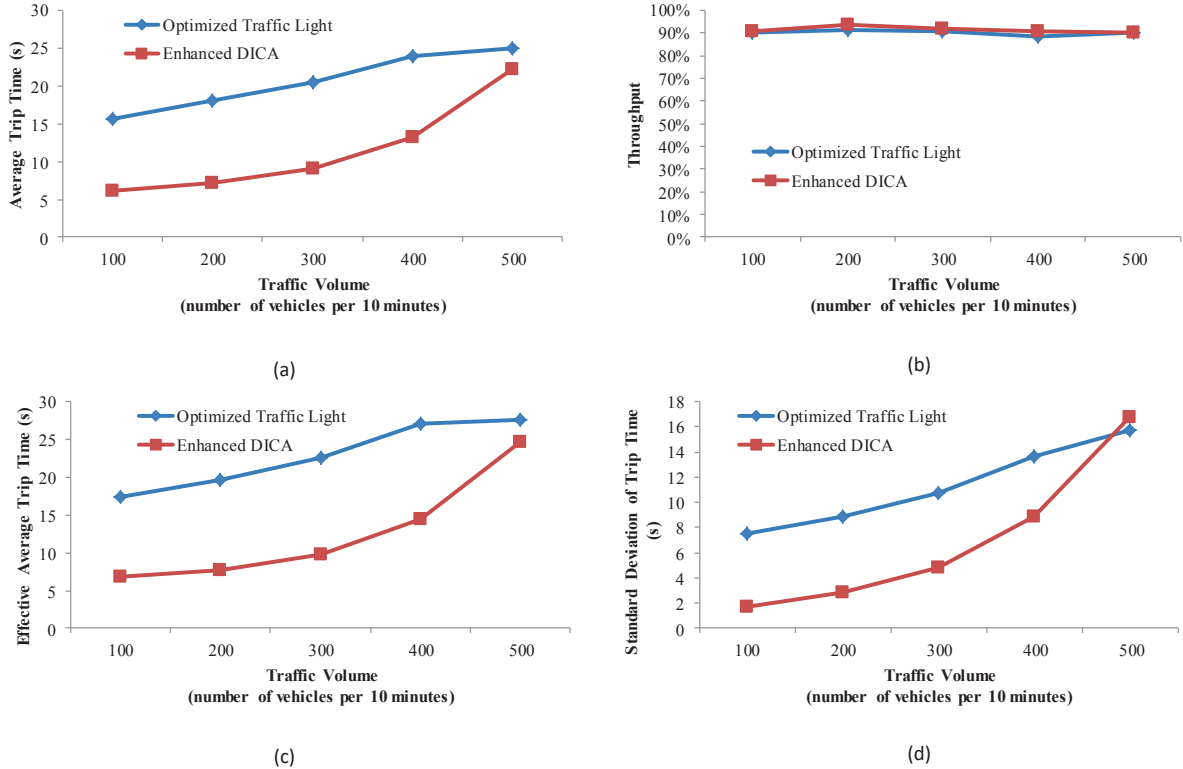}
	\caption[Performance comparison between enhanced DICA and optimized traffic light]{Performance comparison between enhanced DICA and optimized traffic light: (a) average trip time, (b) throughput, (c) effective average trip time, (d) standard deviation of trip time.}
	\label{fig:performance}
\end{figure}

We note that this result is mainly due to the fundamental difference between individual vehicle based traffic coordination algorithms and traffic flow based coordination algorithms. To see this, we can consider a heavy traffic situation when all incoming roads are congested. In such a situation, we know that most vehicles start to cross an intersection at rest when they are allowed to cross the intersection either by green light under traffic light algorithm or confirmation under the proposed DICA. Under a traffic light control, if a vehicle is crossing an intersection, then it is highly likely that a few more following vehicles can also cross the intersection without being stopped. However, in the case when vehicles are controlled by an individual vehicle based coordination algorithm like our enhanced DICA, it is possible to have a situation where vehicles from different roads are permitted alternatively to cross an intersection, which inevitably results in more frequent stops than the case of traffic light control. This is the reason why the enhanced DICA is performing worse and closer to the optimized traffic light in the heavy traffic volume situation.
In fact, this result reveals an important point that, to achieve the best throughput performance, it is necessary to combine both strategies: an individual vehicle based coordination in normal traffic volume and a traffic flow based coordination in congested situation.  
According to this result, we are currently developing algorithms that incorporate the advantage of traffic flow based algorithms when congested into the proposed enhanced DICA.    

Another simulation was performed to validate the transient traffic control performance of DICA when the traffic volume is changing. We run a simulation with 20 minutes long simulation time during which the traffic volume increases from the case of 100 vehicles to 500 vehicles per 10 minutes. 
At each simulation time step, the ratio of the vehicle number generated to the number of vehicles that have exited the intersection, which we call the \emph{flow rate ratio}, was calculated to see how much congestion can occur and also how long it takes to address the congestion. 
The flow rate ratio measured during the simulation time is plotted in Figure \ref{fig:ChangeProb}. In this figure, if the flow rate ratio is close to 1, then it means that all vehicles approached to an intersection have already crossed the intersection and there are no vehicles waiting to cross at that time.  
The simulation time starts from 300s in the figure since the flow rate ratio needs some time to be stable. From the figure, we can also see that before the increase of the traffic volume, the flow rate ratios of the two algorithms are very similar. After 600s at when the traffic volume is changed to the 500 vehicles case, the flow rate ratio of the optimized traffic light increased a lot. Figure \ref{fig:ChangeProb} shows that DICA is more resilient to the change of traffic volume than the optimized traffic light.

\begin{figure}[!t]
	\centering
	\includegraphics[width=3.3in]{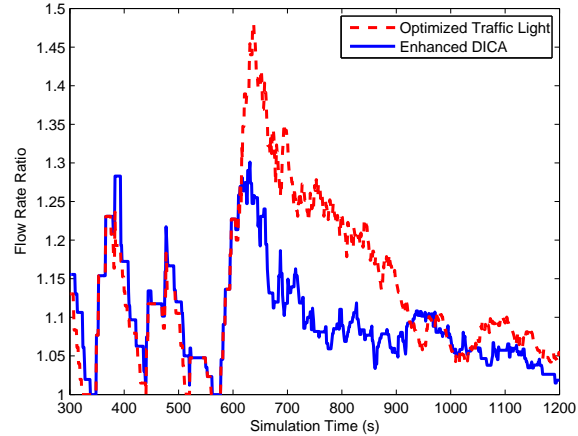}
	\caption{Flow rate ratio when traffic volume changes from 100 to 500.}
	\label{fig:ChangeProb}
\end{figure}

\section{Conclusion} \label{sec:conc}

In this paper, we first introduced our algorithm developed for autonomous and connected intersection traffic management, which is called the discrete-time occupancies trajectory (DTOT) based intersection traffic coordination algorithm (DICA). Subsequently, we showed that the original DICA is deadlock free and also starvation free. We then analyzed the computational complexity of the original DICA and enhanced the algorithm so that it can have better overall computational efficiency. Simulation results show that the computational efficiency of the algorithm is improved significantly after the enhancement and the properties of starvation free and safety are guaranteed. We also validated that the overall throughput performance of our enhanced DICA is better than that of an optimized traffic light control mechanism in case when the traffic is not congested. 
Currently, it is in-progress to integrate the grouping strategy used in traffic flow based intersection control mechanisms into our DICA to achieve the best throughput performance in all traffic volume situations.
We are also working on enhancing the algorithm to deal with sudden emergence of special vehicles such as emergency ambulances or police cars that have the highest priority in real traffic through efficient usage of intersection space. In the future, assumptions like perfect communication, accurate prediction of DTOT will be relaxed and methods to deal with car failures will be studied to make the algorithm more applicable to real situations.
As one of future works, DICA will be generalized to work with mixed traffic where autonomous vehicles and human-driven vehicles coexist.


%
\bibliographystyle{IEEEtran}
\bibliography{computation_complexity}

\begin{thebibliography}{10}
\providecommand{\url}[1]{#1}
\csname url@samestyle\endcsname
\providecommand{\newblock}{\relax}
\providecommand{\bibinfo}[2]{#2}
\providecommand{\BIBentrySTDinterwordspacing}{\spaceskip=0pt\relax}
\providecommand{\BIBentryALTinterwordstretchfactor}{4}
\providecommand{\BIBentryALTinterwordspacing}{\spaceskip=\fontdimen2\font plus
\BIBentryALTinterwordstretchfactor\fontdimen3\font minus
  \fontdimen4\font\relax}
\providecommand{\BIBforeignlanguage}[2]{{%
\expandafter\ifx\csname l@#1\endcsname\relax
\typeout{** WARNING: IEEEtran.bst: No hyphenation pattern has been}%
\typeout{** loaded for the language `#1'. Using the pattern for}%
\typeout{** the default language instead.}%
\else
\language=\csname l@#1\endcsname
\fi
#2}}
\providecommand{\BIBdecl}{\relax}
\BIBdecl

\bibitem{horowitz2000control}
R.~Horowitz and P.~Varaiya, ``Control design of an automated highway system,''
  \emph{Proceedings of the IEEE}, vol.~88, no.~7, pp. 913--925, 2000.

\bibitem{darpa}
\BIBentryALTinterwordspacing
DARPA, ``The darpa urban challenge,'' 2007. [Online]. Available:
  \url{http://archive.darpa.mil/grandchallenge/}
\BIBentrySTDinterwordspacing

\bibitem{bengler2014three}
K.~Bengler, K.~Dietmayer, B.~Farber, M.~Maurer, C.~Stiller, and H.~Winner,
  ``Three decades of driver assistance systems: Review and future
  perspectives,'' \emph{IEEE Intelligent Transportation Systems Magazine},
  vol.~6, no.~4, pp. 6--22, 2014.

\bibitem{markoff2010google}
J.~Markoff, ``Google cars drive themselves, in traffic,'' \emph{The New York
  Times}, vol.~10, no.~A1, p.~9, 2010.

\bibitem{chen}
L.~W. Chen and C.~C. Chang, ``Cooperative traffic control with green wave
  coordination for multiple intersections based on the internet of vehicles,''
  \emph{IEEE Transactions on Systems, Man, and Cybernetics: Systems}, vol.~PP,
  no.~99, pp. 1--15, 2016.

\bibitem{konar}
A.~Konar, I.~G. Chakraborty, S.~J. Singh, L.~C. Jain, and A.~K. Nagar, ``A
  deterministic improved q-learning for path planning of a mobile robot,''
  \emph{IEEE Transactions on Systems, Man, and Cybernetics: Systems}, vol.~43,
  no.~5, pp. 1141--1153, Sept 2013.

\bibitem{li2003autonomous}
T.-H. Li and S.-J. Chang, ``Autonomous fuzzy parking control of a car-like
  mobile robot,'' \emph{IEEE Transactions on Systems, Man, and Cybernetics-Part
  A: Systems and Humans}, vol.~33, no.~4, pp. 451--465, 2003.

\bibitem{mammeri}
A.~Mammeri, D.~Zhou, and A.~Boukerche, ``Animal-vehicle collision mitigation
  system for automated vehicles,'' \emph{IEEE Transactions on Systems, Man, and
  Cybernetics: Systems}, vol.~46, no.~9, pp. 1287--1299, Sept 2016.

\bibitem{colombo}
A.~Colombo and D.~Del~Vecchio, ``Efficient algorithms for collision avoidance
  at intersections,'' in \emph{Proceedings of the 15th ACM international
  conference on Hybrid Systems: Computation and Control}.\hskip 1em plus 0.5em
  minus 0.4em\relax ACM, 2012, pp. 145--154.

\bibitem{le2015autonomous}
S.~Le~Vine, A.~Zolfaghari, and J.~Polak, ``Autonomous cars: The tension between
  occupant experience and intersection capacity,'' \emph{Transportation
  Research Part C: Emerging Technologies}, vol.~52, pp. 1--14, 2015.

\bibitem{onieva2015multi}
E.~Onieva, U.~Hern{\'a}ndez-Jayo, E.~Osaba, A.~Perallos, and X.~Zhang, ``A
  multi-objective evolutionary algorithm for the tuning of fuzzy rule bases for
  uncoordinated intersections in autonomous driving,'' \emph{Information
  Sciences}, vol. 321, pp. 14--30, 2015.

\bibitem{lee2012development}
J.~Lee and B.~Park, ``Development and evaluation of a cooperative vehicle
  intersection control algorithm under the connected vehicles environment,''
  \emph{IEEE Transactions on Intelligent Transportation Systems}, vol.~13,
  no.~1, pp. 81--90, 2012.

\bibitem{wu2012cooperative}
J.~Wu, A.~Abbas-Turki, and A.~El~Moudni, ``Cooperative driving: an ant colony
  system for autonomous intersection management,'' \emph{Applied Intelligence},
  vol.~37, no.~2, pp. 207--222, 2012.

\bibitem{dorigo1996ant}
M.~Dorigo, V.~Maniezzo, and A.~Colorni, ``Ant system: optimization by a colony
  of cooperating agents,'' \emph{IEEE Transactions on Systems, Man, and
  Cybernetics, Part B (Cybernetics)}, vol.~26, no.~1, pp. 29--41, 1996.

\bibitem{malikopoulos2016decentralized}
A.~A. Malikopoulos and C.~G. Cassandras, ``Decentralized optimal control for
  connected and automated vehicles at an intersection,'' \emph{arXiv preprint
  arXiv:1602.03786}, 2016.

\bibitem{dresner2008multiagent}
K.~Dresner and P.~Stone, ``A multiagent approach to autonomous intersection
  management,'' \emph{Journal of artificial intelligence research}, vol.~31,
  pp. 591--656, 2008.

\bibitem{li2013modeling}
Z.~Li, M.~Chitturi, D.~Zheng, A.~Bill, and D.~Noyce, ``Modeling
  reservation-based autonomous intersection control in vissim,''
  \emph{Transportation Research Record: Journal of the Transportation Research
  Board}, no. 2381, pp. 81--90, 2013.

\bibitem{jin2012advanced}
Q.~Jin, G.~Wu, K.~Boriboonsomsin, and M.~Barth, ``Advanced intersection
  management for connected vehicles using a multi-agent systems approach,'' in
  \emph{Intelligent Vehicles Symposium (IV), 2012 IEEE}.\hskip 1em plus 0.5em
  minus 0.4em\relax IEEE, 2012, pp. 932--937.

\bibitem{wuthishuwong2015safe}
C.~Wuthishuwong, A.~Traechtler, and T.~Bruns, ``Safe trajectory planning for
  autonomous intersection management by using vehicle to infrastructure
  communication,'' \emph{EURASIP Journal on Wireless Communications and
  Networking}, vol. 2015, no.~1, pp. 1--12, 2015.

\bibitem{carlino2013auction}
D.~Carlino, S.~D. Boyles, and P.~Stone, ``Auction-based autonomous intersection
  management,'' in \emph{16th International IEEE Conference on Intelligent
  Transportation Systems (ITSC 2013)}.\hskip 1em plus 0.5em minus 0.4em\relax
  IEEE, 2013, pp. 529--534.

\bibitem{vasirani2012market}
M.~Vasirani and S.~Ossowski, ``A market-inspired approach for intersection
  management in urban road traffic networks,'' \emph{Journal of Artificial
  Intelligence Research}, vol.~43, pp. 621--659, 2012.

\bibitem{azimi2013reliable}
S.~Azimi, G.~Bhatia, R.~Rajkumar, and P.~Mudalige, ``Reliable intersection
  protocols using vehicular networks,'' in \emph{Cyber-Physical Systems
  (ICCPS), 2013 ACM/IEEE International Conference on}.\hskip 1em plus 0.5em
  minus 0.4em\relax IEEE, 2013, pp. 1--10.

\bibitem{kim2014mpc}
K.-D. Kim and P.~R. Kumar, ``An mpc-based approach to provable system-wide
  safety and liveness of autonomous ground traffic,'' \emph{IEEE Transactions
  on Automatic Control}, vol.~59, no.~12, pp. 3341--3356, 2014.

\bibitem{kowshik2011provable}
H.~Kowshik, D.~Caveney, and P.~Kumar, ``Provable systemwide safety in
  intelligent intersections,'' \emph{IEEE transactions on vehicular
  technology}, vol.~60, no.~3, pp. 804--818, 2011.

\bibitem{lu2016intelligent}
Q.~Lu and K.-D. Kim, ``Intelligent intersection management of autonomous
  traffic using discrete-time occupancies trajectory,'' \emph{Journal of
  Traffic and Logistics Engineering Vol}, vol.~4, no.~1, pp. 1--6, 2016.

\bibitem{krajzewicz2012recent}
D.~Krajzewicz, J.~Erdmann, M.~Behrisch, and L.~Bieker, ``Recent development and
  applications of sumo--simulation of urban mobility,'' \emph{International
  Journal On Advances in Systems and Measurements}, vol.~5, no. 3\&4, 2012.

\bibitem{cheng2005development}
D.~Cheng, Z.~Z. Tian, and C.~J. Messer, ``Development of an improved cycle
  length model over the highway capacity manual 2000 quick estimation method,''
  \emph{Journal of transportation engineering}, vol. 131, no.~12, pp. 890--897,
  2005.

\bibitem{HCM2000}
T.~R. Board, \emph{Highway Capacity Manual}, National Academy of Sciences,
  Transportation Research Board, Washington, DC, 2000.

\end{thebibliography}
%

%



\end{document}